\documentclass[12pt]{amsart}
\usepackage{amscd,amssymb} %% took out righttag option JV
\usepackage{stmaryrd}
\usepackage{pb-diagram}
\usepackage{multirow}
\usepackage{youngtab} 
\usepackage{version}
\usepackage{lscape}
\usepackage{color}
\usepackage{hyperref}
\usepackage[T1]{fontenc}
\usepackage{paralist}
\usepackage{setspace}

\hypersetup{linktocpage}
\hypersetup{
    colorlinks,
    citecolor=black,
    filecolor=black,
    linkcolor=blue,
    urlcolor=red
}

\newtheorem{thm}{Theorem}[section]

\newtheorem{lemma}[thm]{Lemma}
\newtheorem{prop}[thm]{Proposition}
\theoremstyle{definition}
\newtheorem{defin}[thm]{Definition}

\theoremstyle{remark}
\newtheorem{remark}[thm]{Remark}
\newtheorem{example}[thm]{Example}
\numberwithin{equation}{section}
\font\smc=cmcsc10 at 12.5pt
\font\smallsmc=cmcsc10 
\def\span{\operatorname{span}}

\def\rddots{\cdot^{\cdot^{\cdot^{}}}}

\def\oskip{\par\vbox to4mm{}\par}
\def\R{{\Bbb R}} \def\Z{{\Bbb Z}} \def\CP{{\Bbb{CP}}}
\def \v{{\bf v}}\def \u{{\bf u}}\def \exp{{\rm exp}} \def \log{{\rm log}}\def \r{{\bf r}}\def \k{{\bf k}}
\def\M{{\mathcal M}}  \def \g{{\mathfrak g}}
\def \Ad{{\rm Ad}} \def \ad{{\rm ad}}  \def\A{{\mathcal A}}   
 
  \def\L{{\mathcal L}}    \def\S{{\Bbb S}}
\def\1{{\Bbb I}}
\def\bysame{$\raise.2em\hbox to 3em{\hrulefill}$\thinspace, }

\makeatletter  

\@addtoreset{equation}{section}

    \long\def\symbolfootnote[#1]#2{\begingroup%
    \def\thefootnote{\fnsymbol{footnote}}\footnote[#1]{#2}\endgroup} 
\begin{document}
\thispagestyle{empty}
\oskip
\begin{center}
 \textbf{\textsc{On classification of discrete, scalar-valued Poisson brackets}}
\par\oskip
{\smc e. parodi}\let\thefootnote\relax\footnotetext{\ SISSA, Via Bonomea 265, 34136, Trieste, Italy. E-mail: {\tt parodi@sissa.it}}
\end{center}
\par\oskip
\begin{quote}\footnotesize
{\smallsmc Abstract.} We address the problem of classifying discrete differential-geometric Poisson brackets (dDGPBs) of any fixed order on target space of dimension 1. It is proved that these Poisson brackets (PBs) are in one-to-one correspondence with the intersection points of certain projective hypersurfaces. In addition, they can be reduced to cubic PB of standard Volterra lattice by discrete Miura-type transformations. Finally, improving a consolidation lattice procedure, we obtain new families of non-degenerate, vector-valued and first order dDGPBs, which can be considered in the framework of admissible Lie-Poisson group theory.\\

\end{quote}
\vspace{-0.5cm} 
\begin{spacing}{0.9}
\tableofcontents
\end{spacing}
\vspace{-1.1cm}
\section{Introduction}
In this paper we deal with the following class of local Poisson brackets (PBs) 
\begin{equation} \label{dDGPB}
\begin{array}{cclr}
	\left\{ u^i_n, u^j_{n+k}\right\}_M &=& g^{ij}_k(\u_n, \ldots, \u_{n+k}),	&   0 \leq k  \leq M,\\
	\left\{ u^i_n, u^j_{n+k}\right\}_M &\equiv& 0, 					& k>M,
\end{array}
\end{equation}
defined on the phase space of infinite sequences
$$
\begin{array}{cl}
\u:& \Z \longrightarrow \M^N\\
    & n \longmapsto \u_n \doteq \left( u^i_n\right)_{i=1, \ldots, N}
\end{array}
$$
with values in the target manifold $\M^N$ of dimension $N$. The integer number $M$, called the order of the PB, can be seen as {\it the locality radius}, i.e. the radius of the maximum local interaction between neighboring lattice variables. These PBs have been introduced by B. Dubrovin in~\cite{Dub89} (see also A. Ya. Mal'tsev, \cite{Mal97}), as a discretization of the differential geometric Poisson brackets (DGPBs), defined on the loop space $\L(\M^N) \doteq \left\{ \S^1 \rightarrow \M^N\right\}$ by the formula
\begin{eqnarray}\label{DGPB}
	\left\{ u^i(x), u^j(y)\right\}_M = \sum_{k=0}^M g^{ij}_k(\u(x), \u_x(x), \ldots, \u^{(k)}(x)) \delta^{(M-k)}(x-y), 
\end{eqnarray}
where $i,j = 1, \ldots, N$ and the functions $g^{ij}_k$ are graded-homogeneous polynomials. See for details the papers by S. Novikov and B. Dubrovin~\cite{DubNov83}, \cite{DubNov84}.\\ 
PBs of type \eqref{dDGPB} are associated with lattice Hamiltonian equations of the following form
$$
\dot u^i_n = \left\{ u_n^i, H[\u]\right\}_{M}= \sum \limits_{m\in \Z}   \sum\limits_{p=1}^{N} \left\{ u_n^i, u^p_{m}\right\}_{M}  \dfrac{\strut \delta H[\u]}{\strut \delta u^p_{m}},
$$
where $H[\u] = \sum_{m \in \Z} h(\u_m, \ldots, \u_{m+K})$ for some integer $K\geq0$ and the function $h$ is defined on a finite interval of the lattice. In addition, we define the formal variational derivative as
$$
\frac{\delta H[\u]}{\delta u^p_{n}}  \doteq   \frac{\partial}{\partial u^p_{n}}  \left( 1 + T^{-1} + \ldots+ T^{-K} \right) h(\u_n, \ldots, \u_{n+K}),
$$
where $T$ is the standard shift operator, satisfying
\begin{equation}\label{T}
T^r h(\u_n, \ldots, \u_{n+K}) = h (\u_{n+r}, \ldots, \u_{n+r+K}),
\end{equation}
for any integer $r$. The local Poisson structures of many fundamental integrable systems, such as the Volterra lattices, the Toda lattices, the Bogoyavlensky lattices (see Yu.B. Suris, \cite{Sur}) belong to the class \eqref{dDGPB}. However, the theory of such discrete PBs is much less developed than the corresponding of DGPBs \eqref{DGPB} (see  the survey of O. Mokhov \cite{Mok98} and references therein).\\
A classification of first order ($M=1$) PBs \eqref{dDGPB} has been provided in \cite{Dub89}, whereas it seems that the higher order PBs have not been studied yet. Moreover, to the best of our knowledge, in the literature there are no examples of PBs \eqref{dDGPB} of order $M>2$.\\
\noindent In~\cite{Dub89}, a correspondence between the following first order PBs
\begin{equation}\label{dDGPB1}
\begin{array}{rl}
	\left\{ u^i_n, u^j_{n+1}\right\}_1 =& g_1^{ij}(\u_n, \u_{n+1})\\
	\left\{ u^i_n, u^j_{n}\right\}_1 =& g_0^{ij}(\u_n),
\end{array}
\end{equation} and certain Lie-Poisson groups was discovered. More precisely, if the matrix $g_1^{ij}$ is non-singular (i.e. $\det g_1^{ij} (\u_n, \u_{n+1}) \neq 0$) the PBs \eqref{dDGPB1} are induced by admissible Lie-Poisson group structures on the target manifold $\M^N$ (see Definition \ref{admi} below).\\
Performing a consolidation lattice procedure, which is obtained by defining new varia\-bles of a larger target manifold by the formulas $v^{i+p}_n \doteq u^i_{nM+p}$, $p=0,\ldots, M-1$, one can reduce any PB \eqref{dDGPB} to the form \eqref{dDGPB1}. However, this procedure leaves some unsolved questions:
 {\it
\begin{itemize}
\item[(i)] what are the relations between PBs \eqref{dDGPB} of order $M>1$ and admissible Lie-Poisson groups associated to their consolidations?
\item[(ii)] how to produce examples of such admissible Lie-Poisson groups? 
\end{itemize}
}
\noindent In the present paper, in order to give some partial answers, we classify scalar-valued ($N=1$) PBs \eqref{dDGPB} of any positive order $M$,
\begin{equation}\label{PB-11}
\left\{ u_n, u_{n+k}\right\}_M = g_k(u_n, \ldots, u_{n+k}), \qquad 1\leq k  \leq M.
\end{equation}
\noindent First, we observe that PBs of type \eqref{dDGPB} are invariant under local change of variable
\begin{equation}\label{loc}
u^i_n \longmapsto v^i_n = v^i (\u_n), \qquad i=1, \ldots, N,
\end{equation}
where the coefficients $g^{ij}_k$ transform according to the formula
$$
 g^{ij}_k (\u_n, \ldots, \u_{n+k}) \longmapsto \sum \limits_{p,q=1}^N\frac {\partial v^i}{\partial u^p_n} (\u_n) g^{pq}_k (\u_n, \ldots,  \u_{n+k}) \frac {\partial v^j} {\partial u^q_{n+k}} (\u_{n+k}).
$$
Two local PBs will be therefore considered equivalent if they can be related by a change of coordinates of type \eqref{loc}.\\
\noindent Let us be more precise about the classification of scalar-valued PBs \eqref{PB-11}. The coefficients $g_{k}(u_n, \ldots, u_{n+k})$ satisfy the system of $M^2$ bi-linear PDEs imposed by the Jacobi identity (see Section \ref{Classif}). It turns out that any PB \eqref{PB-11} is characterized by his leading order function $g_M(u_n, \ldots, u_{n+M})$, according to the following results.

\begin{lemma} For any PB of the form \eqref{PB-11}, there exist a set of coordinates {\rm (canonical coordinates)} and an integer $\alpha >0$, such that the leading order reduces to the form
$$
g_{M}(u_{n}, \ldots, u_{n+M})  = f^{\xi} (u_{n+\alpha}, \ldots, u_{n+\alpha+\xi}), \qquad \xi \doteq M-2\alpha,
$$ 
where the function $f^\xi$ is either constant or given by the formula
$$
f^\xi_n \doteq f^\xi (u_{n}, \ldots, u_{n+\xi}) =  \exp \left(  z_{n} \right).
$$
Here,
\begin{equation}\label{z-1}
z_{n} =  \sum_{i=0}^{\xi} \tau_{i} \, u_{n+i},
\end{equation}
and the parameters $\tau_{i}$, $i=0, \ldots, \xi$ satisfy a system of homogeneous polynomial equations (see below Theorem \ref{Thm-2}).
\end{lemma}
\noindent In the case when the leading coefficient $g_M(u_n, \ldots, u_{n+M})$ is constant, it is not difficult to prove that all the other coefficients $\{g_k\}_{k=1, \ldots, M-1}$ are also constant. The general, non-constant case is described by the following  
\begin{thm}\label{Thm-1}
The coefficients $g_k$ of a non-constant PB \eqref{PB-11} are given, in the canoni\-cal coordinates, by suitable linear combinations of the shifted generating function $f^\xi$, according to the following formula
\begin{equation}\label{g_k-1}
g_{\alpha+\xi+p} (u_n, \ldots, u_{n+\alpha+\xi + p}) = \left( \sum_{s= {\rm max} (0, p)}^{{\rm min}(\alpha+p, \alpha)} \lambda^{s}_{p} \, T^{s} \right) f^\xi(u_{n}, \ldots, u_{n+\xi}), 
\end{equation}
where $p=-\alpha, \ldots, \alpha$, and the constants $\lambda$'s can be expressed explicitly in terms of the para\-meters $\tau$'s (see equation \eqref{a-lambda},below). 
\end{thm}
\noindent This Theorem provides a complete classification of PBs of type \eqref{PB-11}. Note that the functional form of the coefficients $g_k$ is fixed by the choice of a finite number of parameters $\tau$.\\
\noindent In addition to the above results, we prove a Darboux-type theorem for PBs \eqref{PB-11}. This is done by considering the change of variable \eqref{z-1}, that is a generalization of the local one \eqref{loc}. Notice that  the formula \eqref{z-1} can be thought as a discrete analogue of Miura transformations, studied in the Hamiltonian PDEs theory (see \cite{DubZha01}). \\
Splitting all the variables into $\alpha+\xi$ families according to 
\begin{equation}\label{c}
v^{(p)}_n \doteq z_{(\alpha+\xi)(n-1)+p}, \qquad p=1, \ldots, \alpha+\xi,
\end{equation}
by direct computation, we obtain that any non-constant PB \eqref{PB-11} in the $z$-coordinates can be reduced to the following simple form
$$
\begin{array}{lcl}
  \left\{ v^{(p)}_n, v^{(p)}_{n+2}\right\} &= &\tau_0\, \tau_\xi\, \exp (v^{(p)}_{n+1})\\
  \left\{v^{(p)}_{n}, v^{(p)}_{n+1}\right\}& = & \tau_0\, \tau_\xi\,  \left[\exp (v^{(p)}_{n})+  \exp (v^{(p)}_{n+1})\right],
  \end{array} 
 $$
that are $\alpha+\xi$ copies of cubic Volterra PB (see Theorem \ref{Thm-3} below).\\
\noindent We consider next compatible pairs of PBs of type \eqref{PB-11}. Recall that a pair of PBs $(P_1, P_2)$ is said to be compatible  (or to form a {\it pencil of PBs}) if any linear combination with constant coefficients $\mu \, P_1 + \nu P_2$ is also a PB. This notion, first mentioned by F. Magri \cite{Magri79} and extended by I. Gel'fand and I. Dorfman \cite{GelfDorf79} (see also \cite{Dorf88}) provides a fundamental device for the integrability of Hamiltonian equations. In our setting, the study of compatible pair of PBs \eqref{PB-11} might lead to the classification of bi-Hamiltonian lattice equations of type
$$
\dot u_n = F(u_{n-S}, \ldots, u_n, \ldots, u_{n+S})=\left\{ u_n, H_1[u]\right\}_{M_1}= \left\{ u_n, H_2[u]\right\}_{M_2}
$$
for some integer $S \geq 0$ and local PBs of order $M_1$, $M_2$. It seems that higher order lattice equations have not been studied yet except for the Volterra type equations (i.e. $S=1$), analyzed by R.I. Yamilov and collaborators using the master symmetries approach (see review article \cite{Yam06} and references therein). 
\noindent The following result describes some necessary conditions for the classification of pencil of PBs \eqref{PB-11}. We expect these conditions to be also sufficient.

\begin{thm}
Let us consider a pair $(P,P')$ of non-constant PBs of type \eqref{PB-11} and order $M=2 \alpha+\xi$ and $M'= 2 \alpha' + \xi'$ respectively, with $M\geq M'$. Then $P$ and $P'$ form a pencil of PBs only if $\alpha= \alpha'$ and there exist coordinates on the manifold $\M^N$ such that
$$
\begin{array}{rcl}
g_M(u_n, \ldots, u_{n+M}) &=& f^\xi (u_{n+\alpha}, \ldots, u_{n+\alpha+\xi}) = \sigma_M \exp \left(\sum_{i=0}^{\xi} \tau_{i} \, u_{n+i}\right)\\
g_{M'}(u_n, \ldots, u_{n+M'}) &=& {f'}^{\xi'} (u_{n+\alpha}, \ldots, u_{n+\alpha+\xi'})= \sigma_{M'} \exp \left( \sum_{i=0}^{\xi'} \tau'_{i} \, u_{n+i}\right)
\end{array}
$$
where $\sigma_M$, $\sigma_{M'}$ are some constants and $\tau_{p} = \tau'_{p} = \tau_{\xi-\xi'+p},$ for any $p=0,\ldots, \xi'$.
\end{thm}

\noindent The structure of the present paper is the following: in Section 2 we describe our classification procedure, proving formula \eqref{g_k-1} for the coefficients $g_k$. In Section 3, using the change of coordinates \eqref{z-1}, suggested by the leading coefficient $g_M$, we prove that any PB of type \eqref{dDGPB} can be reduced to the cubic PB of Volterra lattice (Theorem \ref{Thm-3}). In Section 4, we present some necessary conditions for the compatibility of $(\alpha,\xi)$-brackets (Theorem \ref{Thm-4} and Lemma \ref{L-const}).
Finally, with Section 5, we introduce the concept of admissible Lie-Poisson groups and describe how to produce some new classes of  non-degenerate vector-valued $(N>1)$ first order dDGPBs (Theorem \ref{Thm-5}).

% Classification of (alpha, xi)-brackets

\section{Classification of $(\alpha, \xi)$-brackets}\label{Classif}
This Section is devoted to the classification of the scalar-valued PBs given by the formulas
\begin{equation} \label{PB-1}
\begin{array}{cclcr}
	\left\{ u_n, u_{n+k}\right\}_M &=& g_k(u_n, \ldots, u_{n+k}), & & 1\leq k  \leq M.
\end{array}
\end{equation}
The coefficients $g_{k}(u_n, \ldots, u_{n+k})$ are {\it locally analytic functions} (see R. Yamilov \cite{Yam06} for the precise definition), satisfying, for all values of the independent variables $u_n$, $n \in \Z$, the following bi-linear PDEs given by Jacobi identity
\begin{equation*}\label{J1}\tag*{[p, q]}
\left\{\left\{ u_n, u_{n+p}\right\}, u_{n+p+q}\right\}= \left\{u_n, \left\{u_{n+p}, u_{n+p+q}\right\}\right\} + \left\{\left\{ u_n,u_{n+p+q}\right\}, u_{n+p}\right\},
\end{equation*}
explicitly,
$$
\begin{array}{c}
			   
			      \sum_{i=0}^ p g_p (u_n, \ldots, u_{n+p})_{, u_{n+p-i}}  g_{q+i} (u_{n+p-i}, \ldots, u_{n+p+q})+\\ 
			     - \sum_{i=0}^q  g_q (u_{n+p}, \ldots, u_{n+p+q})_{, u_{n+p+i}} g_{p+i} (u_{n}, \ldots, u_{n+p+i}) \\
			      \shortparallel\\
			      \sum_{i=0}^{p} g_{p+q} (u_n, \ldots, u_{n+p+q})_{, u_{n+i}} g_{p-i} (u_{n+i}, \ldots, u_{n+p})+ \\
			     - \sum_{i=0}^{q} g_{p+q} (u_n, \ldots, u_{n+p+q})_{, u_{n+p+i}} g_{i} (u_{n+p}, \ldots, u_{n+p+i})             
\end{array}                                
$$
where $p, q = 1, \ldots, M$ and $g_0 (\cdot) \equiv 0$, $g_{k} (\cdot) \equiv 0$, $k>M$.

\begin{example}{[{\it Volterra lattice or discrete KdV equation}]}\\
\noindent The well-known Volterra lattice (VL) \cite{Volterra} is defined by the following equations of motion
$$
\dot u_n = u_n \left( u_{n+1} - u_{n-1}\right) \qquad n \in \Z.
$$
Performing the local change of variable $u_n \longmapsto \log\; u_n $, we obtain equations
\begin{equation}\label{VL-eq}
\dot u_n = \exp (u_{n+1}) - \exp (u_{n-1}),
\end{equation}
that admit the following bi-hamiltonian representation (see~\cite{FadTak86})
$$
\dot u_n = \left\{ u_n, H_1[\u]\right\}_1 =  \left\{ u_n, H_2[\u]\right\}_2, 
$$
where
\begin{equation}\label{bi-VL}
\begin{array}{lrcl}
 H_1(u) = \sum_{k} \exp(u_k)                                                             & \left\{ u_n, u_{n+1}\right\}_1 &=&1,\\
\multirow{2}{*}{$H_2 (u) = \frac{1}{2} \sum_{k} u_k$}     &  \left\{ u_n, u_{n+2}\right\}_2 &= & \exp(u_{n+1}) \\
                                                                                                               &  \left\{ u_n, u_{n+1}\right\}_2 &= & \exp(u_{n}) + \exp(u_{n+1}).
 \end{array}
\end{equation}
\noindent These PBs belong to the class \eqref{PB-1}. Moreover, the coefficients of the cubic PB $\left\{ \, \cdot \,, \, \cdot \, \right\}_2$ are characterized by suitable linear combination of the leading order term, given by the exponential functions
$
f^0_n \doteq f^0 (u_n) = \exp(u_n).
$ 
This will be the typical behavior of non-constant PBs \eqref{PB-1}.
\end{example}

\subsection{The leading-order coefficient.}
By the definition of locality radius, the leading-order term $g_M$ might depend on variables $u_n, \ldots, u_{n+M}$, i.e.
$
g_M =g_M(u_n, \ldots, u_{n+M}).
$

%lemma a
\begin{lemma}\label{l-1} 
There exist {\rm canonical coordinates} and an integer $\alpha >0$ such that the leading term $g_M$ reduce to the form
$$ 
g_M(u_n, \ldots, u_{n+M})= f^\xi_{n+\alpha} \doteq f^\xi(u_{n+\alpha}, \ldots,  u_{n+\alpha+\xi}), \qquad  \xi \doteq M-2\alpha,
$$
where
\begin{equation}\label{ast}
\frac{\partial f^\xi_{n}}{\partial u_{n}} \,  \frac{\partial f^\xi_{n}}{\partial u_{n+\xi}} \neq 0
\end{equation}
if $f^\xi$ is a non-constant function.
\end{lemma}
\vspace{-0.7cm}
\begin{proof}
Let us consider the bi-linear PDEs $\ref{J1}$, with $p,q =1, \ldots, M$. At first, we focus our attention on equation 
\begin{equation*}\tag*{$[M, M]$}
\left\{\left\{ u_n, u_{n+M}\right\}, u_{n+2M} \right\} = \left\{u_n, \left\{u_{n+M}, u_{n+2M} \right\}\right\}
\end{equation*}
that provides us
$$
\log\, g_M(u_n, \ldots, u_{n+M})_{, u_{n+M}} = \log\, g_M(u_{n+M}, \ldots, u_{n+2M})_{, u_{n+M}} = \hat a (u_{n+M}) 
$$
for some arbitrary function $\hat a(u_{n+M})$. Solving this logarithmic equation, we obtain the following factorization
$$
g_M(u_n, \ldots, u_{n+M}) = a(u_{n}) f^{M-\alpha-\beta}(u_{n+\alpha}, \ldots, u_{n+M-\beta}) a(u_{n+M}) 
$$
where $\log \, a(u_n)_{, \, u_n} =  \hat a(u_n)$ and $f_{n+\alpha}^{M-\alpha-\beta} =  f^{M-\alpha-\beta}(u_{n+\alpha}, \ldots, u_{n+M-\beta})$ is a constant function or such that
$$
 \frac{\partial f_{n+\alpha}^{M-\alpha-\beta}}{\partial u_{n+\alpha}} \, \frac{\partial  f_{n+\alpha}^{M-\alpha-\beta}}{\partial u_{n+M-\beta}} \neq 0,
$$ for some integers $\alpha, \beta \geq 1$.\\
Performing a local change of the variables $u_n\longmapsto \tilde u_n =  \varphi (u_n)$, we can reduce to 
$$
g_M(\tilde u_n, \ldots, \tilde u_{n+M})  =   f^{M-\alpha-\beta}(\tilde u_{n+\alpha}, \ldots, \tilde u_{n+M-\beta}).
$$

\noindent Finally, we prove $\alpha=\beta$. Indeed, on the one hand, from equation
\begin{equation*}\tag*{$[M, M-\beta]$}
\left\{\left\{ u_n, u_{n+M}\right\}, u_{n+2M- \beta} \right\} = \left\{u_n, \left\{u_{n+M}, u_{n+2M-\beta} \right\}\right\}
\end{equation*}
after some elementary computations, we have
$$
\begin{array}{c}
\log\, g_M(u_{n+\alpha}, \ldots, u_{n+M- \beta})_{, u_{n+M-\beta}} \\
\shortparallel\\
g_{M-\beta}(u_{n+M}, \ldots, u_{n+2M-\beta})_{, u_{n+M}} g_M(u_{n+M-\beta+\alpha}, \ldots, u_{n+ 2(M-\beta)})^{-1}
\end{array}
$$
where the variables appearing on the left-hand side do not intersect with those appearing on the right-hand side. Then, there exists a non-zero constant $k$, such that 
$$
g_{M-\beta}(u_{n+M}, \ldots, u_{n+2M-\beta})_{, u_{n+M}} = k \, g_M(u_{n+M-\beta+\alpha}, \ldots, u_{n+ 2(M-\beta)})
$$
and this PDE makes sense only if  $\alpha \geq \beta$.\\
On the other hand, repeating the same argumentations for equation
\begin{equation*}\tag*{$[M-\alpha, M]$}
\left\{\left\{ u_{n}, u_{n+M-\alpha}\right\}, u_{n+2M-\alpha} \right\} = \left\{u_{n}, \left\{u_{n+M-\alpha}, u_{n+2M-\alpha} \right\}\right\}
\end{equation*}
that is
$$
\begin{array}{c}
g_{M-\alpha} (u_{n}, \ldots, u_{n+M-\alpha})_{, u_{n+M-\alpha}} g_M(u_{n+\alpha}, \ldots, u_{n+ M-\beta})^{-1}\\
\shortparallel\\
 \log\, g_{M}(u_{n+M}, \ldots, u_{n+2M-\beta-\alpha})_{, u_{n+M}} 
\end{array}
$$
we arrive at $\alpha \leq \beta$, that implies $\alpha = \beta$.
\end{proof} 
\noindent For any fixed order $M$, the leading order functions, given by Lemma~\ref{l-1}, define essentially different classes of PBs \eqref{PB-1}.

% ALPHA-XI brackets
\begin{defin}\label{alphaPBs}
Let $(\alpha, \xi)$ be a pair of non-negative integers with $\alpha\geq1$, we call $\left( \alpha, \xi \right)$-brackets the class of non-constant PBs \eqref{PB-1} of order $M= 2 \alpha+ \xi$ expressed in the canonical coordinates, i.e. $g_M(u_n, \ldots, u_{n+M})= f^\xi(u_{n+\alpha}, \ldots,  u_{n+\alpha+\xi})$.
\end{defin}
\noindent For the constant case $f\equiv\sigma_M$ one can immediately prove the following

\begin{prop}\label{const-P}
If there exists a set of coordinates reducing the leading order term to the constant form $g_M(u_n, \ldots, u_{n+M})= \sigma_M$, for some non-zero constant $\sigma_M$,  then all the coefficients $g_k$, $k=1,\ldots, M$ are constant in such coordinates, i.e.
$$
\left\{ u_n, u_{n+k}\right\}_M = g_k (u_n, \ldots, u_{n+k}) = \sigma_k, \qquad k=1, \ldots, M,
$$
where $\sigma_k$ are complex constants.
\end{prop}

\subsection{Classification theorem} 
\begin{thm}\label{Thm}
For any PB \eqref{PB-1} of order $M$, there exist a set of coordinates and an integer $\alpha\geq1$, such that the coefficients $g_{k}(u_n, \ldots, u_{n+k})$, $k=1,\ldots, M$ are given by linear combination of the suitably shifted  function $f^\xi$ (see Lemma \ref{l-1})
\begin{equation}\label{g^alpha_M}
f^\xi (u_{n}, \ldots, u_{n+\xi}) = \exp \left( \sum_{i=0}^{\xi} \tau_{i} \, u_{n+i}\right),
\end{equation}
where $\left\{\tau_{i}\right\}_{i=0, \ldots, \xi}$ are constrained complex parameters. Explicitly, we have
\begin{equation}\label{g_k}
g_{\alpha+\xi + p} (u_n, \ldots, u_{n+\alpha+\xi + p}) = \left( \sum_{s= {\rm max} (0, p)}^{{\rm min}(\alpha+p, \alpha)} \lambda^{s}_{p} \, T^{s} \right) f^\xi(u_{n}, \ldots, u_{n+\xi}), 
\end{equation}
where $\xi \doteq M-2\alpha$, $p=-\alpha, \ldots, \alpha$, $T$ is the shift operator \eqref{T} and $\lambda^{s}_{s-r}$ are scalars such that
\begin{itemize}
%\item[-] are non-zero only if $p = 0, \ldots, \alpha$ and $k=\xi, \ldots, M$,
\item[(i)] they satisfy the multiplication rule
\begin{equation}\label{multi} 
 \lambda^{s}_{s-r} =  \lambda^{\alpha}_{\alpha-r} \,  \left( \lambda^{\alpha}_0 \right)^{-1} \,   \lambda^{\alpha}_{\alpha-s}\qquad s,r=0, \ldots, \alpha,
 \end{equation}
\item[(ii)] denoting $\theta \doteq {\rm min} (\alpha, \xi)$, they can be expressed by explicit formulas in terms of the $\theta+1$ parameters $\tau_0, \tau_1, \ldots, \tau_{\theta-1}; \tau_{\xi}$
\begin{equation}\label{a-lambda}
\begin{array}{lcl}
 \lambda^{\alpha}_{\alpha-r} &=& \det \A_r \left( \{(-)^s \tau_0^{-1} \, \tau_{s}\}_{s\geq0}\right) \quad r=0, \ldots, \alpha-1\\
 \lambda^{\alpha}_{0} &=& \tau_0 \, (\tau_{\xi})^{-1}  
 \end{array}
\end{equation}
where $\A_r$ is the band-Toeplitz matrix
\begin{equation}\label{T1}
\A_{r} \left( {\bf a} \right) = \left( 
\begin{array}{ccccc}
a_1 & a_0 & 0 & \ldots & 0\\
a_2 & a_{1} &a_0 & \ddots &\vdots\\
\vdots& \vdots&\vdots&\ddots&0\\
a_{r-1}&   \ldots &\ldots& a_{1}& a_{0}\\
a_{ r}& \ldots & \ldots &a_{2}& a_{1}
\end{array}
\right)
\end{equation}
associated to the sequence ${\bf a} = (a_0, a_1, \ldots )$.
 \end{itemize}
\end{thm}
\noindent Here and below we adopt the notation
$ \tau_{p} \equiv 0$ if $ p<0$ or $p>\xi$.

\begin{remark}
The scalars appearing on formula \eqref{g_k} can be easily visualized looking at the rows of the following rhombus
\begin{equation}\label{rhombi}
\begin{array}{ll ll lll lll}
g_{(\alpha+\xi)+\alpha}&\longrightarrow \qquad&&&&&\lambda^{\alpha}_{\alpha} \\																					
g_{(\alpha+\xi)+\alpha-1}&\longrightarrow \qquad&&&&\lambda^{\alpha-1}_{\alpha-1}& \lambda^{\alpha}_{\alpha-1}\\
%g_{(\alpha+\xi)+\alpha-2}&\longrightarrow \qquad&&&&\lambda^{\alpha-2}_{\alpha-2}   &  \lambda^{\alpha-1}_{\alpha-2}& \lambda^{\alpha}_{\alpha-2}\\
\qquad\vdots&&&&\rddots&\vdots&\vdots\\
g_{(\alpha+\xi)+1}&\longrightarrow \qquad&&\lambda^{1}_{1}  & \ldots & \lambda ^{\alpha-1}_{1} &\lambda^{\alpha}_{1}\\
g_{(\alpha+\xi)}&\longrightarrow \qquad&\lambda^{0}_{0} & \lambda^{1}_{0}  & \ldots & \lambda ^{\alpha-1}_{0} & \lambda^{\alpha}_{0}\\
g_{(\alpha+\xi)-1}&\longrightarrow \qquad&\lambda^{0}_{-1} & \lambda^{1}_{-1} & \ldots & \lambda ^{\alpha-2}_{-1} &\ \lambda^{\alpha-1}_{-1} &\\
%g_{(\alpha+\xi)-2}&\longrightarrow \qquad&\lambda^{0}_{-2} &\ldots  & \ldots & \lambda ^{\alpha-2}_{-2} & &\\
\qquad \vdots&&\ \ \vdots&\ \ \vdots &\rddots&&&\\
g_{(\alpha+\xi)-\alpha+1}&\longrightarrow \qquad&\lambda^{0}_{-\alpha+1}&  \lambda^{1}_{-\alpha+1} &&&&\\
g_{(\alpha+\xi) -\alpha}&\longrightarrow \qquad&\lambda^{0}_{-\alpha} &&& &&
\end{array}
\end{equation}
\end{remark}
\noindent We complete the description above, specifying the non-trivial constraints that the Jacobi identity imposes on parameters $\tau$ in the following
\begin{thm}\label{Thm-2}
Let $(\alpha, \xi)$ be a pair of non-negative integers and define the sequence $\{\sigma_q\}_{q \geq 0} =\{ (-)^q \lambda^\alpha_{\alpha-q}\}_{q \geq 0}$. Then, the $(\alpha,\xi)$-brackets are in one-to-one correspondence with the intersection points of the $\theta$ projective hypersurfaces in $\CP^\theta$ defined by the homogeneous polynomial equations 
\begin{itemize}
\item[(i)] if $\xi \geq 2\alpha-1$, for any $p=0,\ldots, \alpha-1$,
\begin{equation}\label{cp1}
(\tau_{0})^{p+1} \det \A_{\xi-p} \left( \{\sigma_q\}_{q\geq 0}\right)  = (\tau_\xi)^{p+1} \det \A_p \left( \{ \sigma_{\alpha-q} \}_{q\geq 0}\right) 
\end{equation}
\item[(ii)] if $\xi < \alpha$, for any $p=0, \ldots, \xi-1$, 
\begin{equation}\label{cp2}
(\tau_{0})^{\xi-p} \tau_p =( \tau_\xi)^{\xi-p+1} \det \A_{\xi-p} \left(  \{ \sigma_{\alpha-q}\}_{q\geq 0} \right) 
\end{equation}
\item[(iii)] if $ \alpha \leq \xi < 2\alpha-1$, for any  $p=0, \ldots, \xi-\alpha$
\begin{equation}\label{cp3}
(\tau_{0})^{p+1} \det \A_{\xi-p} \left( \{\sigma_q\}_{q\geq 0}\right)  = (\tau_\xi)^{p+1} \det \A_p \left( \{ \sigma_{\alpha-q} \}_{q\geq 0}\right)  
\end{equation}
and, for any $p= \xi-\alpha+1, \ldots, \alpha-1$
\begin{equation}\label{cp4}
(\tau_{0})^{\xi-p} \tau_p =( \tau_\xi)^{\xi-p+1} \det \A_{\xi-p} \left(  \{ \sigma_{\alpha-q}\}_{q\geq 0} \right).
\end{equation}
\end{itemize}
\end{thm}
\noindent The case $\theta=0$ is described below in the Example \ref{(a,0)}.
\subsection{Proof}
In our computations we often encounter generalized Fibonacci sequences. Therefore, it might be useful to remind some elementary topics about them.
\begin{remark}\label{fib}
For any positive integer $k\geq2$,  the $k$-generalized Fibonacci sequence $\{F_n\}_{n\geq0}$ is defined by the order $k$ linear homogeneous recurrence relation
\begin{equation}\label{fib-eq}
\begin{array}{rcl}
F_n &=& \sum_{i=1}^{k} (-)^{i+1} a_i \, F_{n-i}, \qquad {\rm for} \ n\geq 1,\\
F_0 &\equiv&1
\end{array}
\end{equation}
for arbitrary coefficients $a_i$, $i=0, \ldots, k$.\\
The generating function for  $\{F_n\}_{n\geq0}$ is given by
$$
\varphi(t) = \sum_{n\geq0} F_n t^n = \frac{1}{1-a_1 t + a_2 t^2 - \ldots + (-)^k a_k t^k}
$$
and, using the theory of lower triangular Toeplitz matrices (see for example \cite{Tre} or \cite{Yang08}), one can express the $k$-generalized Fibonacci numbers in terms of determinants of matrices with entries given by the coefficients of their recurrence equation, that is
\begin{equation}\label{e-det}
F_n = \det (a_{1-i+j})_{1 \leq i, j \leq n} = \det \A_n \left( {\bf a}\right),
\end{equation}
where ${\bf a} = (a_0 \equiv 1, a_1, \ldots, a_k, a_{k+1} \equiv 0, \ldots)$ and $\A$ the band Toeplitz matrix defined by $\eqref{T1}$.
\end{remark}
\noindent Let us consider a PB of type \eqref{PB-1}. According to the Lemma \ref{l-1} we can choose certain coordinates on the target manifold $\M^N$ such that the leading term $g_M$ reduces to the following form 
$
f^\xi_{n+\alpha} =f^\xi (u_{n+\alpha},\ldots, u_{n+\alpha+\xi}).
$ 
In our opinion, the simplest way to obtain the form of coefficients $g_k$, $k=1, \ldots, M$ comes from considering equations 
\begin{equation*}\tag*{$[ M-\alpha-p, M ]$}
\left\{ \left\{ u_n, u_{n+M-\alpha-p} \right\}, u_{n+2M-\alpha-p} \right\} =  \left\{ u_n, \left\{ u_{n+M-\alpha-p} , u_{n+2M-\alpha-p} \right\} \right\} 
\end{equation*}
\begin{equation*}\tag*{$[ M,M-\alpha-p ]$}
\left\{ \left\{ u_n, u_{n+M} \right\}, u_{n+2M-\alpha-p} \right\} =  \left\{ u_n, \left\{ u_{n+M} , u_{n+2M-\alpha-p} \right\} \right\} 
\end{equation*}
for any $p=0,\ldots, \alpha+\xi-1$.
\subsubsection{Formula \eqref{g_k}} We start a detailed analysis of some initial cases (i.e. $p=0,1$), that will suggest us how to organize the general computations. For the sake of simplicity, we suppose $\xi>\alpha$.\\ 
\noindent $p=0$. Equations $[ M-\alpha, M ]$ and $[ M,M-\alpha ]$ yield $\left(\log f^\xi_n\right)_{, u_n} = \tau_0$ and $\left(\log f^\xi_n\right)_{, u_{n+\xi}} = \tau_\xi$ for some non-zero constants $\tau_0$, $\tau_\xi$ (see condition \eqref{ast}) and
$$
\begin{array}{ccl}
g_{\alpha+\xi} (u_{n}, \ldots, u_{n+\alpha+\xi}) &=& \lambda^{0}_{0} \, f^\xi_n +\ldots +  \lambda^{\alpha}_{0}\, f^\xi_{n+\alpha},\\
\end{array}
$$
where dots stay for any arbitrary function depending at the most on the variables $u_{n+1},\ldots,$ $u_{n+\alpha+\xi-1}$ and $\lambda^0_0 \doteq (\tau_0)^{-1} \, \tau_\xi$, $\lambda^\alpha_0 \doteq \tau_0 \, (\tau_\xi)^{-1}.$\\
$p=1$. Analogously, equations $[ M-\alpha -1, M ]$ and $[ M, M-\alpha -1]$ provide us $\left(\log f^\xi_n\right)_{, u_{n+1}} = \tau_1$ and $\left(\log f^\xi_n\right)_{, u_{n+\xi-1}} = \tau_{\xi-1}$ for some constants $\tau_1$, $\tau_{\xi-1}$ and
$$
\begin{array}{c}
		  g_{\alpha+\xi-1} (u_{n}, \ldots, u_{n+\alpha+\xi-1})_{, u_{n+\alpha+\xi-1}} \\
		  \shortparallel\\
		  \tau_0 \, g_{M-1} (u_{n}, \ldots, u_{n+2\alpha+\xi-1})+ \tau_1 \, f^\xi(u_{n+\alpha}, \ldots, u_{n+\alpha+\xi})\\
		  \\
		    g_{\alpha+\xi-1} (u_{n}, \ldots, u_{n+\alpha+\xi-1})_{, u_{n}} \\
		    \shortparallel\\
		    \tau_\xi \, g_{M-1} (u_{n-\alpha}, \ldots, u_{n+\alpha+\xi-1})+ \tau_{\xi-1} \, f^\xi(u_{n-1}, \ldots, u_{n+\xi-1}).
\end{array}
$$
It necessarily follows that the constants $\tau_1$, $\tau_{\xi-1}$ are non-zero and
$$
g_{M-1} =g_{M-1} (u_{n+\alpha-1}, \ldots, u_{n+\alpha+\xi}) = \lambda^{\alpha-1}_{\alpha-1} f^\xi_{n+\alpha-1} + h^{\xi-1}(u_{n+\alpha}, \ldots, u_{n+\alpha+\xi-1}) +  \lambda^{\alpha}_{\alpha-1} f^\xi_{n+\alpha}
$$
with $\lambda^{\alpha-1}_{\alpha-1} \doteq (\tau_{\xi})^{-1} \, \tau_{\xi-1}$ and $\lambda^\alpha_{\alpha-1}\doteq -(\tau_0)^{-1} \, \tau_{1}$.

\begin{remark}\label{R-1}
Notice that the function $h^{\xi-1}_{n+\alpha}$ can be understood as the leading order function of a $(\alpha, \xi-1)$-bracket: $h^{\xi-1}_{n+\alpha} \doteq f^{\xi-1}  (u_{n+\alpha}, \ldots, u_{n+M-\alpha-1})$.\\ 
We decide to forget about the contributions provided by the leading order functions ${f}^{\xi'}$ of lower order PBs (i.e. $M'< M$), postponing to Section \ref{compa} the problem of classifying the compatible pairs of PBs \eqref{PB-1}.
\end{remark}
\noindent According to Remark \ref{R-1}, we solve equations $[ M-\alpha -1, M ]$ and $[ M, M-\alpha -1]$, finding
$$
\begin{array}{ll}
 g_{\alpha+\xi-1} (u_{n}, \ldots, u_{n+\alpha+\xi-1}) =&\lambda^{0}_{-1} \,  f^\xi_{n} + \ldots + \lambda^{\alpha-1}_{-1} \, f^\xi_{n+\alpha-1}	,				
\end{array}
$$
where $\lambda^{0}_{-1} \doteq (\lambda^{\alpha}_0)^{-1} \, \lambda^\alpha_{\alpha-1}$, $\lambda^{\alpha-1}_{-1}\doteq \lambda^\alpha_0 \, \lambda^{\alpha-1}_{\alpha-1}$ and dots stay for any arbitrary function depending at the most on variables $u_{n+1},\ldots, u_{n+\alpha+\xi-2}$.
\noindent Iterating this procedure, one can show that the coefficients $\{g_k\}_{k=1, \ldots, M}$ depend on lattice variables according to formulas
$$
\begin{array}{cl}
%g_{M-k} = g_{M-k}(u_{n+{\rm max} \{0, \alpha-k\}}, \ldots , u_{n+{\rm min} \{M-k, \alpha+\xi\}}),& \qquad k=0, \ldots, M-1
g_{\alpha+\xi+k} = g_{\alpha+\xi+k}(u_{n+k}, \ldots , u_{n+\alpha+\xi}),  & k=0, \ldots, \alpha\\
g_{\alpha+\xi-k} = g_{\alpha+\xi-k}(u_{n}, \ldots , u_{n+\alpha+\xi-k}),    & k=1, \ldots, \alpha+\xi-1
\end{array}
$$
and
$$
f^\xi (u_{n}, \ldots, u_{n+\xi})= \lambda^\alpha_\alpha \, \exp \left( \sum_{i=0}^{\xi} \tau_{i} \, u_{n+i}\right),
$$
where 
\begin{itemize}
\item[(i)] $\lambda^\alpha_\alpha$ is a free multiplicative constant, that can be normalized (i.e. $\lambda^\alpha_\alpha \equiv 1$) choosing a suitable rescaling of the coordinates: $u_n \longmapsto k \, u_n$, for some constant $k$,
\item[(ii)]  $\left\{\tau_{i}\right\}_{i=0, \ldots, \xi}$ are non-zero complex parameters, thanks to condition \eqref{ast}.\\
We denote $\tau_{i} \equiv 0$, if $i<0$ or $i>\xi$.
\end{itemize}
\noindent Moreover, solving the equations $\left[M-\alpha-p, M \right]$ and the symmetric ones $\left[M,  M-\alpha-p \right]$, we obtain the formulas \eqref{g_k} for coefficients $\{g_k\}_{k=1, \ldots, M}$.
\subsubsection{The multiplication rule for constants $\lambda$'s} 
Let us focus our attention on equations $\left[ M-\alpha-p, M \right]_{p=0, \ldots, \alpha+\xi-1}$ (analogous results come from equations $\left[ M, M-\alpha-p \right]$). The constraints on constants $\lambda$'s can be encoded into the following linear system 
\begin{equation}\label{eq-gen}
\left(
\begin{array}{c}
\tau_{\xi}\, \lambda^{\alpha-p}_{-p}\\
0\\
\vdots\\
\vdots\\
0
\end{array}
\right)
=
\left(
\begin{array}{ccccc}
\lambda^{\alpha-p}_{\alpha-p} 		& 0 							    & \ldots  &\ldots&0\\
\lambda^{\alpha-p+1}_{\alpha-p}	&  \lambda^{\alpha-p+1}_{\alpha-p+1} &\ddots  &\ldots &\vdots\\
 \vdots 						& \vdots 						    &\ddots  & 0&\vdots\\
 \lambda^{ \alpha-1}_{\alpha-p} 	&  \lambda^{ \alpha-1}_{\alpha-p+1}     & \ldots & \lambda^{\alpha-1}_{\alpha-1}& 0\\
 \lambda^{\alpha}_{\alpha-p} 		& \lambda^{\alpha}_{\alpha-p+1}	    &  \ldots &  \lambda^{\alpha}_{\alpha-1} & \lambda^{ \alpha}_{\alpha}
\end{array}
\right)
\,
\left(
\begin{array}{c}
\tau_{0}\\
\tau_{1}\\
\vdots\\
\tau_{p-1}\\
\tau_{p}
\end{array}
\right)
\end{equation}
where $\lambda^{s}_{s-r}$ is non-zero iff $s,r=0, \ldots, \alpha$, and recursively on $p$ we have defined
\begin{itemize}
\item[(i)] $\lambda^{\alpha-p}_{-p} \doteq \lambda^\alpha_0 \, \lambda^{\alpha-p}_{\alpha-p}$, 
\item[(ii)] $\lambda^{\alpha-p+s}_{\alpha-p} \doteq - \tau^{-1}_{0} \, \left[ \sum^{s}_{i=1} \tau_{i} \,\lambda^{\alpha-p+s}_{\alpha-p+i}\right],$ $s=1, \ldots, p$.
\end{itemize}
Combining $(i)$ and $(ii)$, by induction on $p$ one can prove the following multiplication rule 
\begin{equation}\label{qmulti}
\lambda^{s}_{s-r} = \lambda^{s}_s \, \lambda^{\alpha}_{\alpha-r}, \qquad s,r =0, \ldots, \alpha,
\end{equation}
that allows us to describe all constants $\lambda$'s in terms of the constants $\{ \lambda^\alpha_{\alpha-s}, \lambda^{\alpha-s}_{\alpha-s}\}_{s=0, \ldots, \alpha}$, appearing on the edges of rhombus \eqref{rhombi}.
% Indeed, for any $p=0, \ldots, 2\alpha$, $q=0, \ldots, \alpha$ 
%$$
%\begin{array}{lcl}
%\lambda^{\alpha-q}_{\alpha-p} &=&-\tau^{-1}_{0} \sum^{p-q}_{i=1} \tau_{i} \,\lambda^{\alpha-q}_{\alpha-p+i} =  - \tau^{-1}_0 \lambda^\alpha_q \,  \sum^{p-q}_{i=1} \tau_{i} \,\lambda^{ \alpha}_{\alpha-p+i} \\
%&=& \lambda^{\alpha-q}_{\alpha-q}\,  \lambda^{ \alpha}_{\alpha-(p-q)} 
%\end{array}
%$$
%that is
%$$
%\tau_{0} \, \lambda^{\alpha-q}_{\alpha-q-(p-q)}= \tau_\xi\,  \lambda^{\alpha}_q \,  \lambda^{ \alpha}_{\alpha-(p-q)}  \Longrightarrow \lambda^{\alpha-q}_{\alpha-q-(p-q)}= \lambda^{\alpha}_q \,  \left( \lambda^{\alpha}_0 \right)^{-1}\, \lambda^{ \alpha}_{\alpha-(p-q)}.
%$$

\subsubsection{Constants $\lambda$'s as function of parameters $\tau$'s}
\noindent We start looking at the last row of the matrix appearing on \eqref{eq-gen}. If $0\leq p < \alpha$, we have the following order $p$ recursive relations
\begin{equation}\label{eq-2}
	\lambda^{ \alpha}_{\alpha-p} = - (\tau_{0})^{-1} \, \sum^{p}_{i=1} \tau_{i} \,\lambda^{ \alpha}_{\alpha-p+i}.
\end{equation}
Denoting $F_{p} \doteq \lambda^{ \alpha}_{\alpha-p} $, $ a_p \doteq (-)^p\tau_{0}^{-1}\, \tau_p$, we recognize the $p$-generalized Fibonacci sequence (see equation \eqref{fib-eq}). Then, according to Remark \ref{fib}, we can express 
\begin{equation*}
 \lambda^{\alpha}_{\alpha-p} \left( \tau_{0}, \ldots, \tau_p\right)= \det \A_p \left( \{(-)^s\tau_{0}^{-1}\, \tau_{s}\}_{s\geq0} \right) \qquad p=0, \ldots, \alpha-1 .
\end{equation*}
To write down the constants $\{\lambda^{\alpha-s}_{\alpha-s}\}_{s=0, \ldots, \alpha}$ as functions of the constants  $\{\lambda^{\alpha}_{\alpha-s}\}_{s=0,\dots, \alpha}$, we look at the equations
\begin{equation*}\tag*{$[ M-q, M-\alpha+q ]_{q=0, \ldots, \alpha}$}
\left\{ \left\{ u_n, u_{n+M-q} \right\}, u_{n+2M-\alpha} \right\} =  \left\{ u_n, \left\{ u_{n+M-q} , u_{n+2M-\alpha} \right\} \right\}.
\end{equation*}
They split into
$$
\begin{array}{cll}
(a) & g_{(\alpha+\xi)+\alpha-q} (u_{n+\alpha-q}, \ldots, u_{n+\alpha+\xi})_{, u_{n+\alpha+\xi}}& = \lambda^\alpha_{\alpha-q} \, f^{\xi}(u_{n+\alpha}, \ldots, u_{n+\alpha+\xi})_{, u_{n+\alpha+\xi}}\\
(b) & g_{(\alpha+\xi)+q} \quad(u_{n+q}, \ldots, u_{n+\alpha+\xi)})_{, u_{n+q}} &= \lambda^{q}_{q} \, f^{\xi} (u_{n+q}, \ldots, u_{n+q+\xi)})_{, u_{n+q}}
\end{array}
$$
that are compatible only if
\begin{equation}\label{eq}
\lambda_{q}^{q} = (\lambda^{\alpha}_0)^{-1} \,  \lambda^{\alpha}_{\alpha-q}.
\end{equation}
Substituting \eqref{eq} in the formula \eqref{qmulti}, we hold the multiplication rule \eqref{multi}.

\subsubsection{Constraints for parameters $\tau$'s}
\noindent  Supposing $\xi>\alpha$, we can analyze when $\alpha \leq p \leq \xi$.  Analogously to \eqref{eq-2}, we have the following linear homogeneous recurrence relation of order $\alpha$
\begin{equation}\label{eq-rec}
\tau_p = - \sum_{q=1}^\alpha \lambda^{\alpha}_{\alpha-q}\tau_{p-q}, \qquad {\rm for\ any\ } p\geq \alpha.
\end{equation}
For any $q=1,\ldots, \alpha$, denoting $a_q \doteq (-)^{q} \lambda^\alpha_{\alpha-q}$, we arrive at
$
\tau_{p} = \sum_{q=1}^\alpha (-)^{q+1}a_q\, \tau_{p-q}
$
that is
\begin{equation}\label{e-det}
\tau_p \left( \tau_{0}, \ldots, \tau_{\alpha-1}; \tau_{\xi}\right) = \tau_{0} \, \det \A_p \left(\{ (-)^{s} \lambda^\alpha_{\alpha-s}\}_{s\geq 0} \right),  \qquad p=\alpha, \ldots, \xi.
\end{equation}

\noindent Finally, when $\xi +1 \leq p <  \alpha+\xi$, we obtain analogue relations with respect to the case  $1\leq p <  \alpha$: the equations $\left[ M-\alpha-p, M \right]$ give us
$$
\sum^{q}_{i=0}  \lambda^{\alpha}_{q-i} \tau_{\xi-i} = 0, \qquad q =1, \ldots, \alpha-1.
$$
Following again Remark \ref{fib}, for any $q=1, \ldots, \alpha-1$, we obtain
\begin{equation}\label{eq-3}
\tau_{\xi-q} \left( \tau_{0}, \ldots, \tau_{\alpha-1}; \tau_{\xi}\right) = \tau_{\xi}\, \det \A_q \left( \{(-)^s (\lambda^{\alpha}_0)^{-1}\, \lambda^\alpha_s \}_{s\geq0} \right) 
\end{equation}

\noindent Observe that formulas \eqref{e-det} and \eqref{eq-3} have to be consistent: this provides us the constraints on parameters $\tau$'s. According to the values of $\alpha$ and $\xi$, we immediately find out the homogeneous polynomial equations given by Theorem \ref{Thm-2}.

\subsubsection{Remaining equations}\noindent Due to the relation \eqref{eq}, the symmetric set of equations $\left[ M, M-\alpha-p\right]_{p=0, \ldots, \alpha+\xi-1}$ gives us the same contraints. \\
Moreover, replacing expression \eqref{g_k} on the other equations $\left[ p,q\right]_{p,q =1, \ldots, M}$, with straightforward computations that generalize the previous ones, we can obtain that the relations for the parameters 
$\tau_{0}, \ldots, \tau_{\theta-1}; \tau_\xi$ describe above are necessary and sufficient.

\subsection{Examples}
Looking at the number of the lattice variables which the leading order function can depend on, we study in detail the minimal and the maximal cases.

\subsubsection {$(\alpha, 0)$-brackets} \label{(a,0)}The leading order coefficient can be choosen of the form
$g_{2 \alpha} = f^0(u_{n+\alpha}) =  \exp (\tau_{0} \, u_{n+\alpha})$ and substituting in the formulas for the constants $\lambda$'s we obtain
$$
\begin{array}{rcl}
  					\left\{ u_n, u_{n+2\alpha}\right\}_M &= &  \exp \left(\tau_{0} \, u_{n+\alpha} \right)\\
					\left\{ u_n, u_{n+ \alpha}\right\}_M & = & \exp \left( \tau_{0} \, u_{n} \right) + \exp \left( \tau_{0} \, u_{n+\alpha} \right)
\end{array}
$$
where $\tau_{0}$ is a free complex constant. Notice that the cubic PB of Volterra lattice \eqref{bi-VL} belongs on the class of $(1,0)$-brackets.

\subsubsection {$(1, \xi)$-brackets} These PBs are given in the canonical coordinates by the formulas
\begin{eqnarray*}
\begin{array}{rcl} 
  \left\{ u_n, u_{n+M}\right\} &= &  \exp(z_{n+1})\\
  \left\{ u_n, u_{n+M-1}\right\}&= &  \tau_{0}^{-1}\, \tau_{\xi} \exp(z_{n})+\tau_{0}\, \tau_{\xi}^{-1} \exp(z_{n+1})\\
  \left\{ u_n, u_{n+M-2}\right\}&= & \exp(z_{n})\ \\ 
   \end{array} 
\end{eqnarray*}
where $z_n \doteq \sum_{i=0}^{\xi}  (- \lambda)^i \tau_0 \, u_{n+i}$ and the pair $(\tau_{0}$, $\tau_{\xi})$ belongs to the set of points of $\CP^1$, described by the following equation (see equation \eqref{cp1})
\begin{equation}\label{p-1}
\tau_0^{\alpha+\xi}+ (-)^{\alpha+\xi} \tau_\xi^{\alpha+\xi}=0 \qquad {\rm where}\quad [\tau_0: \tau_\xi] \in \CP^1,\ \alpha=1.
\end{equation}
\noindent Analogously, the $(\alpha, 1)$-brackets are parametrized by the same points \eqref{p-1}, with $\xi=1$.

\subsubsection {A multi-parameters example: $(2, 2)$-brackets} This is the first non trivial case in which the parameters $\tau$'s are described by the intersection points of certain hypersufaces. From the Theorems \ref{Thm} and \ref{Thm-2}, it follows that the $(2,2)$-brackets are characterized by three parameters $\tau_0, \tau_1, \tau_2$, constrained to satisfy the following systems 
$$
\begin{array}{rcl}
(\tau_0)^2 &=& (\tau_2)^2\\
(\tau_1)^{2} &=& 2 \tau_0 \, \tau_{2}.
\end{array} 
$$
%reduction
\section{A Darboux-type theorem}\label{Darboux}
\noindent In Theorem \ref{Thm} we have proven that, up to local point-wise change of coordinates, the leading order function $g_M $ of any PB \eqref{PB-11} can be considered of the following form
$
g_{M} = f^{\xi}_{n+\alpha}  = \exp \left( \sum_{i=0} ^ \xi \tau_{i} \, u_{n+\alpha+i}\right),
$
for some pair of non-negative integers $(\alpha,\xi)$.  \\
In the present Section, we deal with the problem of reduction of PBs \eqref{PB-1} to a canonical form by more general changes of variable than the local ones. Let us give the following
\begin{defin}
Let $(s_1, s_2)$ be a pair of non negative integers, we call a discrete Miura-type transformation any map of the form $ z_n = \varphi (u_{n-s_1}, \ldots, u_{n+s_2})$ that is a {\it canonical transformation}, i.e. a change of coordinates  preserving the PBs. 
\end{defin}
\noindent Notice that these discrete Miura-type transformations, as the continuous ones, are {\it differential substitutions} (i.e. they depend on $u_{n+1}, u_{n+2}, \ldots$) and therefore, only formally invertible. Defined new coordinates according to
\begin{equation}\label{z-2}
z_n= \sum_{i=0}^\xi \tau_i \, u_{n+i},
\end{equation}
we prove by direct computation the following
\begin{thm}\label{Thm-3}
Any $\left( \alpha,\xi\right)$-bracket is mapped by the Miura-type transformation \eqref{z-2} into the following $\left( \alpha+\xi, 0\right)$-bracket,
\begin{eqnarray}\label{z-PB}
\begin{array}{lcl}
  \left\{ z_n, z_{n+2(\alpha+\xi)}\right\} &= & \tau_{0} \, \tau_{\xi}\, \exp (z_{n+\alpha+\xi})\\
  \left\{ z_n, z_{n+\alpha+\xi}\right\}& = &\tau_{0} \, \tau_{\xi}\, \left[ \exp (z_{n})+  \exp (z_{n+\alpha+\xi})\right].
  \end{array}
\end{eqnarray}
\end{thm}
\begin{remark}
Subdividing all the particles-variables into $\alpha+\xi$ families, according to 
$$
v^{(p)}_n \doteq z_{(\alpha+\xi)(n-1)+p}, \qquad p=1, \ldots, \alpha+\xi,
$$
the PB \eqref{z-PB} can be presented in the following simple form
$$
\begin{array}{lcl}
  \left\{ v^{(p)}_n, v^{(p)}_{n+2}\right\} &= & \tau_{0} \, \tau_{\xi}\,  \exp (v^{(p)}_{n+1})\\
  \left\{v^{(p)}_{n}, v^{(p)}_{n+1}\right\}& = & \tau_{0} \, \tau_{\xi}\,  \left[\exp (v^{(p)}_{n})+  \exp (v^{(p)}_{n+1})\right].
  \end{array}
 $$
Therefore, it splits into $\alpha+\xi$ copies of the cubic Volterra PB \eqref{bi-VL}.\\
Furthermore, any constant PB of order $M$, $\left\{u_{n}, u_{n+k}\right\}= \sigma_{k}$, with $k=1, \ldots, M,$ can be reduced by the lattice splitting $v^{(p)}_n \doteq u_{M(n-1)+p},$
 to the following PB 
\begin{equation}\label{c-PB}
\begin{array}{ccl}
  \left\{v^{(p)}_{n}, v^{(q)}_{n+1}\right\}& = & \sigma_{M+q-p} \\
  \left\{ v^{(p)}_n, v^{(q)}_{n}\right\} &= &\sigma_{q-p},
  \end{array}
 \end{equation}
 where $\sigma_0 \equiv 0$ and $p,q=1, \ldots, M$.
Notice that, if all constants $\sigma_k $ are normalized (i.e. $\sigma_k \equiv 1$, $k=1, \ldots, M$), the PB \eqref{c-PB} becomes the quadratic PB for the {\it Bogoyavlensky lattice} (BL) of order $M$  (see the definition of BL in  Suris \cite{Sur}, Chapter 17).
\end{remark}
\begin{proof}
We are interested on the expression of $\left( \alpha, \xi \right)$-brackets in $z$-coordinates. Let be $p=0,1, \ldots$, we have to compute
\begin{equation*}\label{zz}
\begin{array}{rl}
\left\{z_n , z_{n+p} \right\}  &=  \{  \sum_{i=0}^{\xi} \tau_{i} \, u_{n+i},  \sum_{j=0}^{\xi} \tau_j \,  u_{n+p+j} \}=\\
					     &=  \sum_{i=0}^{\xi} \tau_{i} \,\left[ \sum_{j\in J^+_a (i,p)} \tau_j \,  \{ u_{n+i}, u_{n+p+j}\}+\sum_{j\in J^-_a (i,p)} \tau_j \,  \{ u_{n+p+j}, u_{n+i}\} \right]
					      \end{array}
\end{equation*}
where we have defined the sets of admissible $j$'s,
$$
\begin{array}{l}
J^+_a (i,p) \doteq \left\{0, \ldots \xi\right\} \cap  \{i+\xi-p, \ldots, i+2\alpha+\xi-p\}\\
J^-_a  (i,p) \doteq \left\{0, \ldots \xi\right\} \cap \{i-2\alpha-\xi-p, \ldots, i-\xi-p\}.
\end{array}
$$
\noindent In the following, to avoid some technicalities, we suppose $\xi> \alpha$. At first, we notice that when $p > 2(\alpha+\xi)$,  $\left\{z_n , z_{n+p} \right\}$ vanishes.\\
\noindent Let us start considering the set $J^-_a(i,p)$. It is non-empty only if $i= \xi$ and $p=0$. When $p=0$, $J^+_a(i,p)$ is non-empty only if $i=0$, then our summation vanishes, indeed
$$
\begin{array}{rcl}
\left\{z_n , z_{n} \right\} & = & \{  \sum_{i=0}^{\xi} \tau_{i} \, u_{n+i},  \sum_{j=0}^{\xi} \tau_j \,  u_{n+j} \}\\
			              &= & \sum_{j>i} \tau_{i} \, \tau_j \,    \{ u_{n+i}, u_{n+j} \} - \sum_{i>j} \tau_{i} \, \tau_j \,    \{ u_{n+j}, u_{n+i} \} =0.  				      
 \end{array}
$$
\noindent If $p \geq1$, we can reduce to evaluate 
\begin{equation}\label{zzz}
\begin{array}{rcl}
\left\{z_n , z_{n+p} \right\} & = & \sum_{i=0}^{\xi} \tau_{i} \,\left[ \sum_{j\in J^+_a (i,p)} \tau_j \,  \{ u_{n+i}, u_{n+p+j}\} \right].
\end{array}
\end{equation}
\noindent In the following, we give some details about the complete computation. Enforcing a recursive procedure on $p$, for any $i$ that runs from $0$ to $\xi$, we describe the set of admissible $j$: $J^+_{a}(i,p)$. \\
\noindent {\bf Step 1}: $p= 2(\alpha+\xi)$. Then $J^+_a (i, p) \neq \varnothing$ if and only if $i=\xi$ and
$$
\left\{ z_n, z_{n+2(\alpha+\xi)} \right\} = \tau_{0} \, \tau_\xi\,  \left\{ u_{n+\xi} , u _{n+2(\alpha+\xi)} \right\}=\tau_{0} \, \tau_\xi\, \exp(z_{n+\alpha+\xi}).
$$
\noindent {\bf Step 2}: $p= 2(\alpha+\xi)-q$, $q=1, \ldots, \xi$. In the following table: fixed $p$, we describe the non-empty sets $J^+_{a}(i,p)$, as the index $i$ changes.
$$
\begin{array}{ | c ll  ccc cr |}
\hline
 p &   i: J^+_a(i,p) \neq \varnothing&  			  &		& J^+_a (i,p)&	       	&   &  \\
\hline
\hline
2(\alpha+\xi)-q     							   &  i= \xi       & \{ q-2\alpha &\ldots &\ldots & \ldots &q-1 &q \} \\ 
			      				     		            &\quad \vdots     & \{  \ldots	  &\ldots&\ldots&\ldots&\ldots&\ldots \} \\   
	  		   				     	   	            & i= \xi-q+2\alpha     & \{  \quad	     0 &1&\ldots&\ldots&2\alpha-1 &2\alpha \} \\ 
	  		   				     	   	            &\quad \vdots     & \{  				  &\ddots&\ldots&\ldots&\ldots&\ldots \} \\ 
	  		   				     	   	            & i= \xi-q+\alpha     & \{  				      &&0&\ldots&\alpha-1 &\alpha \} \\ 
	  		   				     	   	            &\quad \vdots     & \{  				  &&&\ddots& \ldots&\ldots \} \\ 
	  		   				     	   	            & i= \xi-q+1     & \{  								  &&&&0 &1 \} \\ 
	  		   				     	   	            & i= \xi-q     & \{  								  &&&& &0 \} \\ 
\hline	  
\end{array}
$$
Looking at sets $J^+_{a}(i,p)$, we distinguish three cases depending on the number of elements $\# J^+_a (i,p)$:
$$
\begin{array}{cclclc}
(i) & &&\# J^+_{a}(i,p)&=& 2\alpha+1\\
(ii) &1& \leq & \# J^+_{a}(i,p) &\leq& \alpha+1\\
(iii)&\alpha+1&< & \# J^+_{a}(i,p) &<& 2\alpha+1
\end{array} 
$$
\noindent (i) {\it When $\# J^+_{a}(i,p)= 2\alpha+1$}, we have
$$
\sum_{i=0}^{\xi}  \tau_{i} \, \sum_{j\in\{i+\xi-p,\ldots, i+2\alpha+\xi-p\}}  \tau_j \,  \{ u_{n+i}, u_{n+p+j}\} = \sum_{i=0}^{\xi}  \tau_{i} \,\sum^{2\alpha}_{t=0}  \tau_{r-t} \{ u_{n+i}, u_{n+i+M-t}\} 
$$
where $r\doteq i+M-p$. Now,
$$
\begin{array}{rl}
\sum^{2\alpha}_{t=0}  \tau_{r-t} \{ u_{n+i}, u_{n+i+M-t}\} &= \sum^\alpha_{s=0} \sum^{s+\alpha}_{t=s}  \tau_{r-t} \lambda^{\alpha-s}_{\alpha-t} f^\xi (u_{n+\alpha-s}, \ldots, u_{n+\alpha+\xi-s})\\
												 &=  \sum^\alpha_{s=0} \lambda^{\alpha-s}_{\alpha-s} f^\xi (u_{n+\alpha-s}, \ldots, u_{n+\alpha+\xi-s}) \, \sum^{s+\alpha}_{t=s}  \tau_{r-t} \lambda^{\alpha}_{\alpha-t+s}  
\end{array}
$$
and
$\sum^{s+\alpha}_{t=s}  \tau_{r-t} \lambda^{\alpha}_{\alpha-t+s} = \sum^{\alpha}_{q=0}  \tau_{r-s-q} \lambda^{\alpha}_{\alpha-q}  \equiv 0,$
according to the recurrence relations \eqref{eq-rec} for the parameters $\tau$'s.\\
\noindent (ii) {\it When} $1 \leq \# J^+_{a}(i,p) \leq \alpha+1$, the set $J^+_{a}(i,p)$ is given by $J^+_{a}(i,p) = \left\{ 0, \ldots, r\right\},$ for some $0 \leq r \leq \alpha$. The summation \eqref{zzz} can be written in the following way
$$
\sum_{i=\xi-q}^{\xi-q+\alpha}  \tau_{i} \,\left[ \sum_{j\in\{0, \ldots, r\}}  \tau_j \,  \{ u_{n+i}, u_{n+p+j}\} \right] \qquad  r \doteq i-\xi+q.
$$
Noticing that $\sum\limits_{j\in\{0, \ldots, r\}}  \tau_j \,  \{ u_{n+i}, u_{n+p+j}\} = \sum\limits_{s=0}^{r} \tau_{r-s}\{ u_{n+i}, u_{n+i+M-s}\} = \tau_{0} \, \lambda^{\alpha-r}_{\alpha-r} f^\xi_{n+i+\alpha-r}$ we obtain
$$
\sum_{i=\xi-q}^{\xi-q+\alpha}  \tau_{i} \,\left[ \sum_{j\in\{0, \ldots, r\}}  \tau_j \,  \{ u_{n+i}, u_{n+p+j}\} \right] =   \left[\sum_{i=\xi-q}^{\xi-q+\alpha} \tau_{i} \, \lambda^{\alpha}_{i-\xi+q}  \right] \, f^\xi_{n+\alpha+\xi-q}
$$
and $\sum_{i=\xi-q}^{\xi-q+\alpha} \tau_{i} \, \lambda^{\alpha}_{i-\xi+q} = \sum_{t=0}^{\alpha} \tau_{\xi-q+t} \,  \lambda^{\alpha}_{t} \equiv 0$, see the recurrence \eqref{eq-rec}.

\noindent (iii) {\it When} $\alpha+1 < \# J^+_{a}(i,p) < 2\alpha+1$, the set $J^+_{a}(i,p)$ is $J^+_{a}(i,p) = \left\{ 0, \ldots, r\right\},$ for some $\alpha < r \leq 2\alpha$. With analogous calculations, one can directly prove that
the summation $\sum_{j\in\{0, \ldots, r\}}  \tau_j \,  \{ u_{n+i}, u_{n+p+j}\} $ vanishes.

\noindent {\bf Step 3}: $p= 2\alpha+\xi-q$, $q=1,\ldots, \alpha+1$.  The sets $J^+_{a}(i,p)$ are given by 
$$
\begin{array}{ | c ll  ccc cr |}
\hline
 p &  i: J^+_a (i,p) \neq \varnothing&  			  &		& J^+_a (i,p)&	       	&   &  \\
\hline
\hline
2\alpha+\xi-q     							   &  i= \xi       & \{ \xi-2\alpha+q &\ldots & \xi & & & \} \\ 
				     		            &\quad \vdots     & \{  	\quad	\ldots		  &\ldots&\ldots&\ddots&& \} \\  
 								           &\quad \vdots     & \{  	\quad	\ldots		  &\ldots&\ldots&\ldots&\ddots& \} \\  
	  		   				     	   	            & i= \xi-q    & \{  \quad	     \xi-2\alpha &\ldots&\ldots&\ldots&\xi-1 &\xi \} \\ 
	  		   				     	   	            &\quad \vdots     & \{  	\quad \ldots			  &\ldots&\ldots&\ldots&\ldots&\ldots \} \\ 
	  		   				     	   	            & i= 2\alpha - q     & \{  				    \quad 0  &\ldots&\ldots&\ldots&\ldots &2\alpha \} \\ 
			    				     	   	            &\quad \vdots     & \{  				  &\ddots&\ldots&\ldots & \ldots&\ldots \} \\ 
	  		   				     	   	            &\quad \vdots     & \{  				  &&\ddots&\ldots & \ldots&\ldots \} \\ 
	  		   				     	   	            & i= 0     & \{  								  &&&0&\dots &q \} \\ 
\hline	  
\end{array}
$$
We find out again the two situations 
$$
\begin{array}{cclclcc}
(i) & && \# J^+_{a}(i,p)&=& 2\alpha+1 & \checkmark\\
(ii) & \alpha+1&<& \# J^+_{a}(i,p) &<& 2\alpha+1& \checkmark
\end{array} 
$$
that we have been already studied in the previous Step. Then the summation \eqref{zzz} vanishes.

\noindent {\bf Step 4}: $p= \alpha+\xi$. In this case, we have
$$
\begin{array}{ | ll ll  ccc ccr |}
\hline
 p & J^+_{i,p} & i: J^+_a \neq \varnothing&  			  &		& J^+_a (i)&	       	&   &&  \\
\hline
\hline
\alpha+\xi     &    \{i-\alpha, \ldots, i+\alpha\}	   &  i= \xi       & \{ \xi-\alpha &\ldots &\ldots &\xi && & \} \\ 
	   		&				     	   	            &\quad \vdots     & \{  	\quad	\ldots		  &\ldots&\ldots&\ldots&\ddots&& \} \\ 
				   		&				     	   	            &\quad \vdots     & \{  	\quad	\ldots		  &\ldots&\ldots&\ldots&\ldots&\ddots& \} \\ 
 	  		   &				     	   	            & i= \xi-\alpha     & \{  \xi-2\alpha   &\ldots&\ldots&\ldots&\ldots&\ldots &\xi \} \\ 
	  		   &				     	   	            &\quad \vdots     & \{  		\vdots		  &\ldots&\ldots&\ldots&\ldots&\ldots&\ldots \} \\ 
	  		   &				     	   	            & i= \alpha    & \{  	\quad	0		      &\ldots&\ldots&\ldots&\ldots& &2\alpha \} \\ 
	  		   &				     	   	            &\quad \vdots     & \{  				 & \ddots&\ldots&\ldots&\ldots& \ldots&\ldots \} \\ 
			   &				     	   	            &\quad \vdots     & \{  				 & &\ddots&\ldots&\ldots& \ldots&\ldots \} \\ 
	  		   &				     	   	            & i= 0     & \{  							&	  &&0&\ldots&\ldots &\alpha \} \\ 
\hline	  
\end{array}
$$
three different cases
$$
\begin{array}{cclclcc}
(i) & && \# J^+_{a}(i,p)&=& 2\alpha+1 & \checkmark\\
(ii) & \alpha+1&<& \# J^+_{a}(i,p) &<& 2\alpha+1& \checkmark\\
(iii) &&& \# J^+_{a}(i,p) &=& \alpha+1&?
\end{array} 
$$
Only the last one provides some contributions. Indeed, 
$$
\begin{array}{rcl}
\sum_{j\in\{0, \ldots, \alpha\}}  \tau_j \,  \{ u_{n+i}, u_{n+p+j}\} &=& \tau_{0} \, \lambda^{0}_{0} f^\xi (u_{n+i}, \ldots, u_{n+i+\xi})
\end{array}
$$
and
$$
\begin{array}{rcl}
\sum_{j\in\{\xi-\alpha, \ldots, \xi\}}  \tau_j \,  \{ u_{n+i}, u_{n+p+j}\} &=& \tau_{\xi}\, \lambda^{\alpha}_{0} f^\xi (u_{n+\alpha+i}, \ldots, u_{n+i+\alpha+\xi}).
\end{array}
$$
Our summation \eqref{zzz} becomes
$$
  \left\{ z_n, z_{n+\alpha+\xi}\right\}= \tau_0 \, \tau_\xi \, f^\xi (z_n) +  \tau_0 \, \tau_\xi \, f^\xi (z_{n+\alpha+\xi}) =  \tau_0 \, \tau_\xi \left[ \exp (z_{n})+  \exp (z_{n+\alpha+\xi})\right].
$$

\noindent Finally, {\bf Step 5}: $p= 0, \ldots, \xi+\alpha-1$ can be analyzed similarly to Step 2 and Step 3. 
\end{proof}

\section{Compatible pairs}\label{compa}
We are now in a position to study the compatible pairs $(P,P')$ of PBs \eqref{PB-1}. We provide some necessary conditions that we expect to be also sufficient. This might be the starting point for a future classification of the still little-understood bi-hamiltonian higher order scalar-valued difference equations. \\

\noindent Let $P$ and $P'$ two PBs of type \eqref{PB-1}. Their leading order functions are given respectively by
\begin{equation}\label{f}
g_M (u_n, \ldots, u_{n+M}) = a_M(u_n)\, f^\xi (u_{n+\alpha}, \ldots, u_{n+\alpha+\xi}) \, a_M(u_{n+M})
\end{equation}
\begin{equation}\label{f'}
g'_{M'} (u_n, \ldots, u_{n+M'}) = a'_{M'}(u_n)\, f^{\xi'}(u_{n+\alpha'}, \ldots, u_{n+\alpha'+\xi' }) \, a'_{M'}(u_{n+M'})
\end{equation}
where $f^\xi_{n+\alpha} = \sigma_M \exp \left( \sum_{p=0}^{\xi} \tau_p \, u_{n+\alpha+p}\right)$ and $f^{\xi'}_{n+\alpha'}= \sigma_{M'} \exp \left( \sum_{p=0}^{\xi'} \tau'_{p} \, u_{n+\alpha'+p}\right)$\\
for some non-zero constant $\sigma_M$ and $\sigma_{M'}$. It is not restrictive to suppose $M\geq M'$.

\begin{lemma}\label{l123}
A pair of non-constant PBs \eqref{PB-1} $(P,P')$ forms a pencil of PBs only if there exists a local change of variables, reducing the leading coefficients \eqref{f} and \eqref{f'} to the formulas
\begin{equation}\label{g}
g_M (u_n, \ldots, u_{n+M}) = f^\xi (u_{n+\alpha}, \ldots, u_{n+\alpha+\xi}) 
\end{equation}
and
\begin{equation}\label{g'}
g'_{M'} (u_n, \ldots, u_{n+M'}) ={f}^{\xi'} (u_{n+\alpha}, \ldots, u_{n+\alpha+\xi'}) 
\end{equation}
for some functions $f^\xi$ and ${f}^{\xi'}$. Notice that $\alpha=\alpha'$.
\end{lemma}

\begin{proof}
Let $P$ and $P'$ PBs defined respectively by leading function \eqref{f} and \eqref{f'}. From equation
\begin{equation*}\tag*{$[M,M']$}
\left\{ \left\{ u_n, u_{n+M'} \right\}, u_{n+M+M'} \right\} =  \left\{ u_n, \left\{ u_{n+M'} , u_{n+M+M'} \right\} \right\}
\end{equation*}
we find out
$
\left( \log\, a_{M} (u_n)\right)_{, {u_n}} = \left( \log\, a'_{M'} (u_n)\right)_{, {u_n}},
$
that implies $a_M (u_n) = k\, a'_{M'}(u_n)$ for some constant $k$. A suitable change of coordinates leads us to \eqref{g} and \eqref{g'}, where the constant $k$ has been absorbed in the multiplicative constant $\sigma_{M'}$. We hold $\alpha= \alpha'$ looking, for example, at equation
\begin{equation*}\tag*{$[M,M'-\alpha]$}
\left\{ \left\{ u_n, u_{n+M} \right\}, u_{n+M+M'-\alpha} \right\} =  \left\{ u_n, \left\{ u_{n+M} , u_{n+M+M'-\alpha} \right\} \right\}
\end{equation*}
that makes sense if and only if $M'-\alpha \leq M'-\alpha'$, i.e. $\alpha \geq \alpha'$. In such case we have
\begin{equation}\label{eq-4}
\begin{array}{c}
{f^\xi_{n+\alpha}}_{, u_{n+\alpha+\xi}}\, f^{\xi'}_{n+\alpha+\xi+\alpha'} = g'_{\alpha+\xi'} (u_{n+M}, \ldots, u_{n+\alpha+\xi+ M'})_{, u_{n+M}}  \, f^\xi_{n+\alpha}
\end{array}
\end{equation}
that gives 
$
g'_{\alpha+\xi'} (u_{n+M}, \ldots, u_{n+M+M'-\alpha})_{, u_{n+M}} =  \tau_\xi \, f^{\xi'}( u_{n+\alpha+\xi+\alpha'}, \ldots, u_{n+\alpha+\alpha'+\xi+\xi'}). 
$
This equation makes sense if $\alpha \leq \alpha'$. It follows $\alpha= \alpha'$, so it is not restrictive to suppose $\xi>\xi'$.
\end{proof}

\begin{thm} \label{Thm-4}
A pair of non-constant PBs \eqref{PB-1} $(P,P')$, defined by leading functions \eqref{g} and \eqref{g'} for certain parameters $\tau$ and $\tau'$ satisfying algebraic constraints of Theorem \ref{Thm-2}, forms a pencil of PBs only if  
\begin{equation}\label{eq-egu}
\tau_{p} = \tau'_p = \tau_{\xi-\xi'+p},
\end{equation}
for any $p=0, \ldots, \xi'$.
\end{thm}

\begin{proof}
Let $P$ and $P'$ be a pair of PBs \eqref{PB-1} respectively, of order $M=2\alpha+\xi$ and $M'= 2 \alpha+\xi'$ (i.e. $M > M'$), leading functions \eqref{g} and \eqref{g'} and suppose that any linear combination
$\mu \, P + \nu \, P'$, for $\mu, \nu $ arbitrary constants, is a PB of order $M$ (i.e. it satisfies the bi-linear PDEs, coming from Jacobi identity).\\ 
Analogously to the procedure followed in the proof of Theorem \ref{Thm}, we find necessary conditions looking recursively at equations
$$
\begin{array}{lll}
\left[ M', M-\alpha-p \right],& \left[ M-\alpha-p,M' \right] & p=0, \ldots, \alpha+\xi'
\end{array}
$$
\noindent When $p=0$, equation 
\begin{equation*}\tag*{$[M', M-\alpha]$}
\left\{ \left\{ u_n, u_{n+M'} \right\}, u_{n+M+M'-\alpha} \right\} =  \left\{ u_n, \left\{ u_{n+M'} , u_{n+M+M'-\alpha} \right\} \right\} 
\end{equation*}
gives us
$$
 \log f^{\xi'}(u_{n+\alpha}, \ldots, u_{n+\alpha+\xi'})_{, u_{n+\alpha+\xi'}} =	\lambda^{0}_{0}\, \log f^\xi (u_{n+M'}, \ldots, u_{n+M'+\xi})_{, u_{n+M'}} 
$$
that enable us to identify $\tau'_{\xi'} \equiv  \tau_\xi.$ Analogously, equation$[M-\alpha, M']$ provides $\tau'_0 \equiv  \tau_{0}$.

\noindent Iterating this procedure, when $1\leq p \leq \alpha$, equations $\left[ M-\alpha-p, M' \right]$ restricted to the function $f^{\xi'}$ originate some constraints that can be organized in the following matrix form
\begin{equation*}
\left(
\begin{array}{c}

\tau_{\xi}\, \lambda'^{\alpha-p}_{-p}\\
0\\
\vdots\\
\vdots\\
0
\end{array}
\right)
=
\left(
\begin{array}{ccccc}
\lambda'^{\alpha-p}_{\alpha-p} 		& 0 							    & \ldots  &\ldots&0\\
\lambda'^{\alpha-p+1}_{\alpha-p}	&  \lambda'^{\alpha-p+1}_{\alpha-p+1} &\ddots  &\vdots&\vdots\\
 \vdots 						& \vdots 						    &\ddots  & \ddots&\vdots\\
 \lambda'^{ \alpha-1}_{\alpha-p} 	&  \lambda'^{ \alpha-1}_{\alpha-p+1}     & \ldots & \lambda'^{\alpha-1}_{\alpha-1} & 0\\
 \lambda'^{\alpha}_{\alpha-p} 		& \lambda'^{\alpha}_{\alpha-p+1}	    &  \ldots &  \ldots &  \lambda'^{\alpha}_{\alpha} 
\end{array}
\right)
\,
\left(
\begin{array}{c}
\tau_{0}\\
\tau_{1}\\
\vdots\\
\tau_{p-1}\\
\tau_p
\end{array}
\right)
\end{equation*}
and recursively on $p$, we obtain that $\lambda'^{\alpha}_{\alpha-p} \equiv \lambda^{\alpha }_{\alpha-p}$, for any $p=0, \ldots, \alpha.$ Adding the contribution coming from equations $\left[ M',M-\alpha-p\right]$, the following constraints on the parameters $\tau$'s hold
$$
\begin{array}{rcl}
\tau_p &=& \tau'_{p}\\
\tau_{\xi-p} &=& \tau'_{\xi'-p}
\end{array} \qquad p=0, \ldots, \alpha-1.
$$
The formula \eqref{eq-egu} immediately follows.
\end{proof}

\subsection{Compatibility with constant brackets}
We devote this sub-section to study the compatibility of pairs $(P,P')$ of PBs \eqref{PB-1}, where $P'$ can be reduced by local change of coordinates to the constant form.

\begin{lemma}\label{L-const}
Any non-constant PB \eqref{PB-1} of order $M=2\alpha+\xi$, with $\xi>0$ is compatible with a constant bracket of order $M'$ only if $M'\leq \alpha$. Moreover, the constant coefficients $\{\sigma_i\}_{i=1, \ldots, M'}$ have to satisfy the following recurrences
\begin{equation}\label{eq-5}
\begin{array}{rl}
\sigma_{\alpha-k} &=  \left( \sum_{i=0}^k  \lambda^{\alpha-i}_{\alpha-k} \right) \sigma_{\alpha},  \quad  k=0, \ldots, \alpha,\\
\sigma_0 &\equiv 0.
\end{array}
\end{equation}
\end{lemma}
\begin{proof}
Let us fix a pair of PBs \eqref{PB-1} $(P,P')$ such that their leading functions are of the form 
$$
\begin{array}{rcl}
g_M(u_n, \ldots, u_{n+M}) &=& a_M(u_n) \, f^\xi (u_{n+\alpha}, \ldots, u_{n+\alpha+\xi}) \, a_M(u_{n+M})\\
g'_{M'}(u_n, \ldots, u_{n+M'}) &=& a_{M'}(u_n) \, \sigma_{M'} \, a_{M'}(u_{n+M'}), \ \sigma_{M'} \neq 0.
\end{array}
$$
As in the proof of Lemma \ref{L-const}, from equation $\left[ M,M'\right]$, we obtain $\left( \log\, a_{M} (u_n)\right)_{, {u_n}} = \left( \log\, a'_{M'} (u_n)\right)_{, {u_n}}$, that is $a_{M} (u_n) = k \, a_{M'}(u_{n})$, for some constant $k$.\\
Then, after a change of variables, we can reduce to consider
$$
\begin{array}{rcl}
g_M(u_n, \ldots, u_{n+M}) &=& f^\xi (u_{n+\alpha}, \ldots, u_{n+\alpha+\xi})\\
g'_{M'}(u_n, \ldots, u_{n+M'}) &=&\sigma_{M'}.
\end{array}
$$
Let us now suppose $\alpha< M' \leq M$, then from the equation
\begin{equation*}\tag*{$[M, M'-\alpha]$}
\left\{ \left\{ u_n, u_{n+M} \right\}, u_{n+M+M'-\alpha} \right\} =  \left\{ u_n, \left\{ u_{n+M} , u_{n+M+M'-\alpha} \right\} \right\}  
\end{equation*}
we obtain
$ f^\xi (u_{n+\alpha}, \ldots, u_{n+\alpha+\xi})_{, u_{n+\alpha+\xi}} \sigma_{M'} = 0 $ that implies $\sigma_{M'} \equiv 0$.\\
\noindent Let us suppose $\xi>0$ and $k=1, \ldots, \alpha$, then $M-k >\alpha$. We focus our attention on equations 
\begin{equation*}\tag*{$[k,M-k]$}
\left\{ \left\{ u_n, u_{n+k} \right\}, u_{n+M} \right\} =  \left\{ u_n, \left\{ u_{n+k} , u_{n+M} \right\} \right\} + \left\{ \left\{ u_n, u_{n+M} \right\}, u_{n+k} \right\}, 
\end{equation*}
obtaining the constraints system $ \lambda^{\alpha-k}_{\alpha-k} \sigma_\alpha = \tau^{-1}_0  \sum_{i=0}^{k} \sigma_{\alpha+i-k}\, \tau_{i} $ that can be written in the form \eqref{eq-5}. \\
\end{proof}
\begin{remark}
When $\xi=0$, we immediately find that any $(\alpha, 0)$-bracket is compatible with a constant PB of order $\alpha$, given by the formulas
$$
\begin{array}{ccl}
\{ u_n, u_{n+\alpha} \} &=& \sigma_\alpha\\
\{ u_n, u_{n+\alpha-s} \} &\equiv& 0, \qquad s=1,\ldots, \alpha-1.
\end{array}
$$

\end{remark}

\subsection{Examples.} We complete this Section, adding the details for the relevant family of $(1,\xi)$-brackets, where we are able to prove that our necessary conditions are also sufficient.
\noindent Let $P$ be a $(1, \xi)$-bracket,
\begin{itemize}
\item[-] if $\tau_{0} \neq \tau_\xi$, according to Theorem \ref{Thm-4}, there are not $(1,\xi')$-brackets $P'$, such that the pair $(P,P')$ forms a pencil of PBs. Looking for constant brackets, we have that $P$ is compatible only with first order constant brackets if and only if $\xi = 0,1$.\\
\item[-] if $\tau_{0} = \tau_\xi$, then necessarily $M=2K$ for some positive integer $K$. Any pair of $(1, 2q)$-brackets, with $q=1, \ldots, K$, defined by building functions of the form
$$
f^{2q} (u_{n+1}, \ldots, u_{n+2q-1}) =  \exp(z^{(2q)}_{n+1}),
$$
where $z^{(2q)}_n \doteq \sum_{i=0}^{2(q-1)} (- \tau_{0})^{i} u_{n+i}$, forms a pencil.
\end{itemize}

\section{On non-degenerate, vector-valued PBs}\label{Vector}
In order to provide some examples of vector-valued PBs, we can define new lattice variables, according to the following 
\begin{prop}\label{p-consoli}
Let us consider a pair of positive integers $(M,K)$, with $K \leq M$. Any scalar-valued PB \eqref{PB-1} of order $M$, according to the formula
$$
v_n^{1+p} \doteq u^1_{nK+p}, \ \ \ p=0, \ldots, K-1,
$$
is transformed into a non-degenerate PB (i.e. the leading order is given by a non-singular matrix) of order $A$ and target space of dim. $K$ iff $M=A \, K$, for some positive integer $A$.
\end{prop}

\noindent We are interested on vector-valued PBs of first order (i.e. $A=1$), then we put $K=M$ in the above proposition. First, let us recall some preliminary definitions.\\
\noindent A Lie group $G$, with a Poisson bracket $\left\{\,\cdot\,,\,\cdot\,\right\}_G$ is a {\it Lie-Poisson group} if the multiplication $\mu: G \times G \rightarrow G$ is a mapping of Poisson manifolds, where on $G \times G$ is defined the bracket $\left\{\varphi, \psi \right\}_{G \times G}(g,h) = \left\{ \varphi\left(\,, h\right), \psi\left(\,, h\right)\right\}_G (g) + \left\{ \varphi\left(g, \,\right), \psi\left(g, \, \right)\right\}_G (h)$, with $h,g \in G.$ Let $c^{k}_{ij}$ be the structure constants of a Lie algebra $\g$. The couple $(\g, \gamma)$ is a {\it Lie bi-algebra} if and only if
\begin{itemize}
		\item[(i)] $\gamma$ is a 1-cocycle on $\g$ with values on $\g \otimes \g$, where $\g$ acts on $\g \otimes \g$ by the adjoint representation $\ad^{(2)}_\xi = \ad_\xi \otimes {\bf 1} + {\bf 1} \otimes \ad_\xi$, that means: $\delta_{\ad} \gamma =0$, i.e.
$$
	\ad^{(2)}_\xi \left(\gamma(\eta)\right)-\ad^{(2)}_\eta \left(\gamma(\xi)\right) - \gamma\left(\left[\xi, \eta\right]\right)=0,
$$
or, fixed a basis of $\g$,  
$
	\ c^\epsilon_{rs} \gamma^{pq}_\epsilon = c^p_{\epsilon s} \gamma^{\epsilon q}_r + c^q_{\epsilon s} \gamma^{p \epsilon}_r - c^p_{\epsilon r} \gamma^{\epsilon q}_s - c^q_{\epsilon r} \gamma^{p \epsilon}_s.
$
	\item[(ii)] $^{t}\gamma: \g^\ast \otimes \g^\ast \rightarrow \g^\ast$ defines a Lie bracket on $\g^\ast$: $\left[\xi, \eta\right]_{g^\ast}= ^{t} \gamma\left(\xi \otimes \eta\right).$ 
\end{itemize}
\noindent The correspondence between the Lie-Poisson groups and the Lie bi-algebras is clarified by the following 

 \begin{thm} Let $G$ a Lie group, with tangent Lie algebra $\g$: locally a Lie-Poisson structure on $G$ is uniquely (up to isomorphism) determined by a Lie algebra structure on the dual space $\g^\ast$, then $\g$ is a Lie bialgebra $\left( \g, \gamma \right)$.
 \end {thm}
 \begin{proof} All the detail of the proof can be found in \cite{KS}. Here we recall some ideas about the direction from Lie bi-algebra to Lie-Poisson group. Fixed a basis of $\left( \g, \gamma \right)$, the constants $\gamma^{p q}_k$ define a Lie bracket on $\g ^\ast$. Moreover,  solving the differential equation $$\gamma^{pq}_k = \partial_k \pi_G ^{p q}|_e,$$ we can define  
 $
 \{\phi, \psi \}_G \doteq \partial_p \phi \pi_G^{p q} \partial_q \psi,
 $
 that is a PBs because the compatibility condition of the system
\begin{equation}\label{pi_G}
 \left\{ \begin{array}{rcl}
            \partial_k \pi_G^{pq}&=& c^{p}_{sk}\pi_G^{s q} + c^{q}_{sk} \pi_G^{ps} + \gamma^{pq}_{k}\\
          \pi_G^{pq}|_e &=&0\\
\end{array} \right.       
 \end{equation}
is guaranteed by $(i)$, i.e. $\gamma$ is a $1-$ cocycle on $\g$.
 \end{proof}
 
 \noindent We focus on the sub-class of Lie Poison group, given by the following
 \begin{defin}\label{admi}
 A Lie-Poisson group $\{G, \{ \,, \, \}_G \}$ is called {\bf admissible}, if there exist:
 \begin{itemize}
 \item[(i)]  A skew-symmetric matrix $\k$$\in \Lambda^2 \g$, such that the cohomologous 1-cocycle $\tilde \gamma$, defined by $\tilde \gamma \doteq \delta_\ad \k + \gamma$, i.e. $\tilde \gamma ^{pq}_t  \doteq \gamma ^{pq}_t + c^p_{st} k^{sq} + k^{ps} c^q_{st}$, provides a Lie algebra structure on $\g^\ast$. Notice that $\k$ has to satisfy the Yang-Baxter equation 
  $$
 k^{sq}  c^p_{st} k^{tr} +  k^{ps} c^q_{st} k^{tr} +  k^{ps} c^r_{sq}  k^{tq} = \gamma^{pq}_s k^{sr} + \gamma^{rp}_s k^{sq} + \gamma^{qr}_s k^{sp}.
$$
 \item[(ii)] A Lie algebra homomorphism $\r$  $:(\g^\ast, \gamma^{pq}_s) \longrightarrow (\g, c^s_{pq})
 $  such that $\r_\ast: r_\ast^{pq} = r^{qp}$
 defines a Lie algebra homomorphism $\r_\ast: \left( \g^\ast, \tilde{\gamma}^{pq}_s\right)\rightarrow \left(\g, c^s_{pq}\right)$.
 \end{itemize} 
 \end{defin}
 
\noindent We are now in a position to formulate the Dubrovin's theorem \cite{Dub89}.
\begin{thm}\label{ThmDub} An admissible Lie-Poisson group $\left(G, \left\{ \,, \, \right\}_G \right)$ together with corresponding matrices $\r$, $\k$ defines a Poisson bracket of the form 
\begin{equation}\label{PB-N1}
\begin{array}{rcl}
           \left\{ u^i_n, u^j_{n}\right\}_1 &=& h^{ij}\left( \u_n \right)\\
             \left\{ u^i_n, u^j_{n+1}\right\}_1 &=& g^{ij} \left( \u_n, \u_{n+1} \right)
\end{array}
\end{equation}
where $\u_n \in G$ for all $n$, according to the following formulas
$$
	\left\{\varphi \left(\u_n\right), \psi \left(\u_{n+1}\right)\right\}_1\doteq \partial_\alpha \varphi\left(\u_n\right) r^{\alpha \beta}\partial'_\beta \psi\left(\u_{n+1}\right)
$$
where $\partial_\alpha$ and $\partial'_\beta$ are left- and right-invariant vector fields on $G$,
$$
	\left\{\varphi \left(\u_n\right), \psi\left(\u_n\right)\right\}_1\doteq h^{\alpha \beta}\left(\u_n\right) \partial_\alpha \varphi\left(\u_n\right)\partial_\beta \psi\left(\u_{n}\right)
$$
where $h^{\alpha \beta}(\u_n) \doteq \pi_G^{\alpha \beta}\left(\u_n\right) + \Ad_{u^{-1}}^{(2)} k^{\alpha \beta}$ and $\pi_G^{\alpha \beta}(\u_n)$ is determined by the system~\eqref{pi_G}. Here $\varphi, \psi$ are arbitrary smooth functions on $G$. Viceversa, all brackets of such form are obtained in this way under the non-degeneracy condition $\det g^{ij} \neq 0$.
\end{thm} 

\begin{remark}
This theorem does not seem to have simple applications in the lattice systems, studied in the literature. For example, one can notice that the fundamental Toda lattice has three well-known local compatible PBs (see \cite{Sur}), but only the quadratic one is non-degenerate and can be represent using the previous theorem on the algebra of ${\rm Aff}^{0} \R^1$, group of affine transformations of the straight line.
\end{remark}

\noindent Our new vector-valued $(N>1)$ PBs \eqref{PB-N1} are provided by the following

\begin{thm}\label{Thm-5}
Any $(\alpha, \xi)$-bracket, after the consolidation lattice procedure, becomes a non-degenerate PB of the form \eqref{PB-N1}, with associated Lie bi-algebra given by
$$
\begin{array}{llll}
\g_{(\alpha,\xi)} = \span \left\{ L_s\right\}_{s=1, \ldots, M},  & \left[ L_p, L_{q} \right] (\v_n) =   c_{p\,q}^r \, L_{r}(\v_n)\\
%\left[ L_p, L_{\alpha+\xi+q} \right] =   \gamma_{p,q}\, L_{p},  & p = 1, \ldots, \alpha+\xi\\
\g^\ast_{(\alpha,\xi)} = \span \left\{ R^s\right\}_{s=1, \ldots, M},  & \left[ R^p, R^q \right](\v_n) =   \gamma^{p\,q}_r \, R^{r}(\v_n), 
\end{array}
$$ 
where the summation is over the repeated index $r$ and $c_{p\,q}^r$ and $\gamma^{p\,q}_r$ are certain constants depending on the parameters $\tau$'s.
\end{thm}
\begin{proof}
Choosing $K=M$ in the Proposition \ref{p-consoli}, we have
$$
g^{ij}(\v_n, \v_{n+1})=\left(
\begin{array}{ccc}
  \left\{ v^1_{n}, v^1_{n+1}\right\}_M &  &   \\   
 \vdots  &\ddots&  \\
  \left\{ v^{M}_n, v^{1}_{n+1}\right\}_M  & \ldots &   \left\{ v^M_{n}, v^M_{n+1}\right\}_M 
  \end{array} \right)
$$
and $ h^{ij}(\v_n) = g^{ji}(\v_{n}, \v_{n}) - g^{ij}(\v_n, \v_{n})$. Substituting the explicit formulas \eqref{g_k}, we find out a decomposition of the leading order
$
g^{ij}(\v_n, \v_{n+1}) = L^{i}_\mu (\v_n) R^{\mu\, j}(\v_{n+1}).
$
The matrix $L(\v_n)$ is given by $L(\v_n) \doteq \Lambda_L \, {\rm diag} (l_1, \ldots, l_M) (\v_n)$ where $l_s(v^{\alpha+s}_n, \ldots, v^{\alpha+\xi+s}_n) = \exp \left(\sum_{i=0}^\xi \tau_{i} \, v^{\alpha+s+i}_n \right)$ and
$$
\Lambda_L \doteq \left(
\begin{array}{lll lll}
\lambda^\alpha_\alpha &						&		&	&&\\
\vdots 			     &\lambda^\alpha_\alpha	&		&	&&\\
\lambda^0_0		     & \vdots				&\ddots	&	&&\\	
				     & \lambda^0_0			&		&\ddots &&\\	
				     &						&\ddots	&	&\ddots&\\
				     &						&		& \lambda^0_0 & \ldots & \lambda^\alpha_\alpha
\end{array}
\right)
$$
Analogously the matrix $R(\v_{n+1})$ factorizes into $R(\v_{n+1}) \doteq {\rm diag}(r_1, \ldots, r_M) \, \Lambda_R$, where $r_s(v^{1}_{n+1}, \ldots, v^{s-\alpha}_{n+1}) = \exp \left(\sum_{i=0}^{\alpha-s+1} \tau_{\xi-i} v^{s-\alpha-i}_{n+1}\right)$,
$$
\Lambda_R \doteq \left(
\begin{array}{lll lll}
\lambda^\alpha_\alpha &						&		&	&&\\
\vdots 			     &\lambda^\alpha_\alpha	&		&	&&\\
\lambda^\alpha_0		     & \vdots				&\ddots	&	&&\\	
				     & \lambda^\alpha_0			&		&\ddots &&\\	
				     &						&\ddots	&	&\ddots&\\
				     &						&		& \lambda^\alpha_0 & \ldots & \lambda^\alpha_\alpha
\end{array}
\right)
$$
Notice that we are using the notation $v^s_n \equiv 0$ if $s<1$ or $s>M$. Moreover, to avoid some technicalities, we suppose $\xi>\alpha$.\\
\noindent Now, let $\{L_s\}_{s=1, \ldots, M}$ be the $s$-th column of matrix $L$ and $\g_{(\alpha, \xi)} \doteq {\rm span} \, \left\{ L_s\right\}$ as vector space.  By direct computation, we find that the following commutators 
\begin{equation} \label{com}
\left[ L_p, L_q\right]^k(\v_n) = L_q^s (\v_n)\, L_{p, s}^k(\v_n) - L_p^s(\v_n)\, L_{q, s}^k(\v_n)  = c_{p\,q}^r \, L_r(\v_n)
\end{equation}
equip $\g_{(\alpha, \xi)}$ of a Lie algebra structure. At first, we observe that if  $\alpha+\xi<q \leq M$, $L_q(\v_n)$ are constant vector fields. Then
$$
\left[ L_p, L_q\right]^k(\v_n) = \sigma_{p\,q} \, L_{p} (\v_n), \qquad p=1, \ldots, M,
$$ 
for some constants $\sigma_{p\,q}$ that can be expressed in terms of parameters $\tau$.

\noindent In particular, observe that 
$\left[ L_s, L_{\alpha+\xi+t}\right]^k(\v_n) = \delta_{s,t} \, L_{s}$, where $s,t=1, \ldots, \alpha$ and $\delta$ is the Kronecker symbol.\\
\noindent Let us now consider $p,q =1, \ldots, \alpha+\xi$. It is not restrictive to suppose $q>p$. Denoting $t\doteq q-p$, the formulas \eqref{com} for the commutators reduce to
$$
\left[ L_p, L_{p+t}\right]^k(\v_n) = L_{p+t}^s (\v_n) \, L_{p, s}^k(\v_n). 
$$
We distinguish two cases:
\begin{itemize} 
\item[(i)]When $ t\leq \alpha$, we have $\left[ L_p, L_{p+t}\right]^k(\v_n) = \lambda^{\alpha-k+p}_{\alpha-k+p}\, \sum_{i=0}^{t} \lambda^{t-i}_{t-i} \tau_{i} \, l_{p+t}(\v_n) l_p(\v_n)$. Let us focus on the summation $\sum_{i=0}^{t} \lambda^{t-i}_{t-i} \tau_{i}$. According to formulas \eqref{multi} and \eqref{eq-rec}, we have 
$$
\sum_{i=0}^{t} \lambda^{t-i}_{t-i} \tau_{i}= \lambda^{0}_0 \, \sum_{i=0}^{t} \lambda^{\alpha}_{\alpha-t+i} \tau_{i} \equiv 0.
$$
\item[(ii)] When $t>\alpha$, we obtain $\left[ L_p, L_{p+t}\right]^k(\v_n) = \lambda^{\alpha-k+p}_{\alpha-k+p}\, \sum_{i=0}^\alpha \lambda^{\alpha-i}_{\alpha-i} \tau_{t-\alpha+i}  \, l_{p+t}(\v_n) l_p(\v_n)$, and we have
$ \sum_{i=0}^\alpha \lambda^{\alpha-i}_{\alpha-i} \tau_{t-\alpha+i}= \lambda^{0}_0 \, \sum_{i=0}^\alpha \lambda^{\alpha}_{i} \tau_{t-\alpha+i} \equiv 0$, according to \eqref{eq-rec}.
\end{itemize}
\noindent Finally, starting from the rows of matrix $R$, we can repeat the procedure above finding the Lie algebra structure on $\g^\ast$.
\end{proof}

%\noindent{\bf Acknowledgements.} 

%\oskip\oskip\footnotesize
%\noindent {\it 1991 Mathematics Subject Classification:}
%Primary 54A99, 58A50, secondary 18F99, 54H20.
%\par\noindent
%{\it Keywords:} Locally ringed superspaces, virtual superspaces, topological superspaces, topological spaces of points.

\end{document}